\documentclass[11pt]{article}
 
\usepackage{amsfonts, amsthm}
\usepackage{amssymb}
\usepackage{amsmath}  
\usepackage{graphicx, comment}
\usepackage{authblk}
\usepackage{mathrsfs}
\usepackage[colorlinks=true, allcolors=blue]{hyperref}
\usepackage{cleveref}

\usepackage[dvipsnames]{xcolor}

\usepackage{fullpage}

\newtheorem{theorem}{Theorem}
\newtheorem{corollary}[theorem]{Corollary}
\newtheorem{lemma}[theorem]{Lemma}

\newtheorem{prop}[theorem]{Proposition}

\theoremstyle{remark}

\newcommand{\ket}[1]{|#1\rangle}
\newcommand{\bra}[1]{\langle#1|}
\newcommand{\ketbra}[2]{|#1\rangle\langle#2|}

\newcommand{\tr}{\text{\rm tr}}

\newcommand{\Ent}{\text{\rm Ent}}

\newcommand{\cE}{\mathcal{E}}

\newcommand{\cK}{\mathcal{K}}

\newcommand{\dd}{\text{\rm{d}}}

\newcommand{\cF}{\mathcal{F}}
\newcommand{\cL}{\mathcal{L}}

\newcommand{\diag}{\text{\rm{diag}}}

\newcommand{\E}{\mathbb{E}}
\newcommand{\bfa}{\mathbf{a}}
\newcommand{\bfq}{\mathbf{q}}
\newcommand{\hatp}{\hat{p}}

\newcommand{\thermal}{\text{\rm{th}}}

\newcommand{\Phiatt}{\Phi^{\text{att}}}
\newcommand{\Phiamp}{\Phi^{\text{amp}}}
\newcommand{\Phiad}{\Phi^{\text{add}}}

\title{A Meta Logarithmic-Sobolev Inequality for Phase-Covariant Gaussian Channels}

\author{Salman Beigi$^{1, 2}$,\, Saleh Rahimi-Keshari$^3$}
\affil{\it \footnotesize $^1$School of Mathematics, Institute for Research in Fundamental Sciences (IPM), P.O. Box 19395-5746, Tehran, Iran \\
\it \footnotesize $^2$Centre for Quantum Technologies, National University of Singapore, Singapore\\
\it \footnotesize $^3$School of Physics, Institute for Research in Fundamental Sciences (IPM), P.O. Box 19395-5531, Tehran, Iran}

\begin{document}

\maketitle

\begin{abstract}
	We introduce a meta logarithmic-Sobolev (log-Sobolev) inequality for the Lindbladian of all single-mode phase-covariant Gaussian channels of bosonic quantum systems, and prove that this inequality is saturated by thermal states. We show that our inequality provides a general framework to derive information theoretic results regarding phase-covariant Gaussian channels. Specifically, by using the optimality of thermal states, we explicitly compute the optimal constant $\alpha_p$, for $1\leq p\leq 2$, of the $p$-log-Sobolev inequality associated to the quantum Ornstein-Uhlenbeck semigroup. Prior to our work, the optimal constant was only determined for $p=1$. Our meta log-Sobolev inequality also enables us to provide an alternative proof for the constrained minimum output entropy conjecture in the single-mode case. Specifically, we show that for any single-mode phase-covariant Gaussian channel $\Phi$, the minimum of the von Neumann entropy $S\big(\Phi(\rho)\big)$ over all single-mode states $\rho$ with a given lower bound on $S(\rho)$, is achieved at a thermal state.	
\end{abstract}

%{\footnotesize\tableofcontents}

%******************************************************************************
\section{Introduction}

Quantum Gaussian channels are prototype noise models for the transmission of quantum information in current quantum communication technologies and information processing. In particular, they model transmission channels for sending
quantum data encoded in bosonic systems through optical fibers and free space. Despite several developments in the past decades, there are still wide open conjectures regarding information theoretic properties of these channels. These conjectures are important not only due to the significance of quantum Gaussian channels in quantum information
science, but also as quantum generalizations of some influential results in functional analysis.

In functional analysis and information theory, the optimality of Gaussian functions and Gaussian distributions has been established for various optimization problems involving Gaussian kernels or Gaussian channels. To name a few results, it is shown in~\cite{CARLEN1991} that only Gaussian functions saturate Gross's celebrated \emph{logarithmic-Sobolev (log-Sobolev) inequality}~\cite{Gross75}. Also, it is shown in~\cite{Lieb1990Gaussian} that ``Gaussian kernels have only Gaussian maximizers." In information theory, it is shown that the optimal input distribution for Gaussian broadcast channels is Gaussian~\cite{GengNair2014}. More generally, various inequalities including Brascamp-Lieb inequalities, the Loomis-Whitney inequality, the Pr\'ekopa-Leindler inequality, sharp form of Young's inequality for convolution of functions as well as the entropy power inequality are tight for Gaussian functions and distributions; see~\cite{AnantharamJogNair2022} and references therein. Thus, it is tempting to verify the validity of these results in the non-commutative case for quantum Gaussian channels. 

Single-mode quantum Gaussian channels have been classified in~\cite{Holevo2007Structure}; see also~\cite{Holevo+2013Book}. Among these classes are the three important physical classes of \emph{attenuator,} \emph{amplifier} and \emph{additive-noise channels}, which have two main features. First, these channels are \emph{phase-covariant} meaning, that they are covariant with respect to the group of phase operators, and second, they form quantum Markov semigroups~\cite{HeinosaariHolevoWolf2010}. These two features make these channels suitable candidates for generalizing some of the aforementioned properties of Gaussian kernels in functional analysis to the quantum case. In particular, as we show in this paper, these channels satisfy a new inequality that is saturated by Gaussian states, which we refer to as \emph{meta log-Sobolev inequality} and provides a general framework for deriving quantum information-theoretic results regarding bosonic systems.

One application of our meta log-Sobolev inequality is to compute the optimal constants of the $p$-log-Sobolev inequalities for a semigroup of attenuator channels, called the \emph{quantum (bosonic) Ornstein-Uhlenbeck semigroup}. This semigroup is of particular interest as it resembles its classical counterpart when acting on certain operators; see~\cite[Equation~(7.5)]{Cipriani+2000}. 
The classical Ornstein-Uhlenbeck semigroup is an important example of a 
Markov semigroup, which can be understood as the convolution of an input distribution by a Gaussian one. Its associated log-Sobolev inequality takes the form
\begin{align}\label{eq:C-LSI}
	\frac 12\Ent_2(f)\leq \|f'\|_{2}^2, \qquad\quad \forall f>0,
\end{align}
where the entropy $\Ent_2(f) = \E[f^2\log f^2] - \E[f^2]\log \E[f^2]$ and the $2$-norm $\|f\|_2=\E[f^2]^{1/2}$ are defined with respect to the standard normal distribution. This inequality was first established by Gross in his seminal work~\cite{Gross75}. 

Given the log-Sobolev inequality~\eqref{eq:C-LSI} for the classical Ornstein-Uhlenbeck semigroup and the development of the theory of hypercontractivity for quantum semigroups~\cite{OZ99}, it is natural to ask if the quantum Ornstein-Uhlenbeck semigroup is hypercontractive and satisfies a log-Sobolev inequality.\footnote{Amplifier and additive-noise channels are not hypercontractive in the usual sense.} The $2$-log-Sobolev inequality for the quantum Ornstein-Uhlenbeck semigroup takes the form
\begin{align}\label{eq:Q-LSI}
	\alpha_2\,\Ent_{2, \sigma_\beta} (X)\leq \tr\Big(\!\sqrt{\sigma_\beta} [\bfa, X]^\dagger \sqrt{\sigma_\beta} [\bfa, X]\Big),\qquad  \forall X>0,\; \tr\Big(\sigma^{\frac 14}_\beta X \sigma_\beta^{\frac 14}\Big)^2\!\!<+\infty.
\end{align}
Here, $\sigma_\beta =(1-e^{-\beta}) e^{-\beta a^\dagger a}$ is a thermal (and then Gaussian) state with parameter $\beta>0$, $[\bfa, X]=\bfa X-X\bfa$ is the commutator of the annihilation operator with $X$, and in the special case that $\rho = \Big(\sigma^{\frac 14}_\beta X \sigma_\beta^{\frac 14}\Big)^2$ is a quantum state
$$\Ent_{2, \sigma_\beta} (X)=D(\rho\| \sigma_\beta)= \tr(\rho\log \rho)- \tr(\rho\log \sigma_\beta).$$
Moreover, $\alpha_2$ is called the $2$-log-Sobolev constant whose optimal value, depending on $\beta$, is one of our main problems of study in this work.

Using the spectral properties of the Lindbladian of the quantum Ornstein-Uhlenbeck semigroup established in~\cite{Cipriani+2000}, a bound on the 2-log-Sobolev constant $\alpha_2$ is derived in~\cite{CarboneSasso2008}. This bound confirms the hypercontractivity of the quantum Ornstein-Uhlenbeck semigroup, yet it seems far from being optimal. Moreover, the bound of~\cite{CarboneSasso2008} works only in the single-mode case and is not clear to satisfy the so called \emph{tensorization property.}

Although $2$-log-Sobolev inequalities (assuming some regularity conditions~\cite{OZ99, KT13}) imply hypercontractivity inequalities, for a more refined characterization of the latter inequalities we need to establish generalizations of the above $2$-log-Sobolev inequality for all values of $1\leq p\leq 2$. These generalizations, called $p$-log-Sobolev inequalities, are all equivalent and ignorant of the parameter~$p$ for the classical Ornstein-Uhlenbeck semigroup, yet they differ in the quantum case.\footnote{The parameter $p$ is also relevant in the study of log-Sobolev inequalities for \emph{finite} classical Markov semigroups; see~\cite{Mossel+2013} and references therein.}

The 1-log-Sobolev (also called the \emph{modified log-Sobolev}) inequality for the quantum Ornstein-Uhlenbeck semigroup is first studied in~\cite{HuberKoenigVershynina17} where it is examined for Gaussian states and it is conjectured that Gaussian states are optimal for the $1$-log-Sobolev inequality. This conjecture is first proven in~\cite{CarlenMaas17}. An alternative proof of this inequality and an extension are also derived in~\cite{DePalmaHuber18}. 

Our meta log-Sobolev inequality is also related to computing the classical capacity of attenuators, amplifiers and additive-noise channels, i.e., finding the maximum rate of reliable transmission of classical data through these channels. Since these channels are phase-covariant, their classical capacity is related to their minimum output entropy, i.e., the minimum entropy of their output states; see, e.g.,~\cite{Holevo+2013Book}. It was a long-standing open problem that the minimum output entropy of these channels is additive and is achieved over Gaussian states~\cite{HolevoWerner2001Evaluating, GiovannettiGuha+2004Minimum, GiovannettiHolevoLloydMaccone2010}. This conjecture was answered in the affirmative in~\cite{GiovannettiHolevoGarcia-Patron2015}. Extensions of this result have also been proven for which we refer to~\cite{Holevo_2015, DePalma+2018survey} and references therein.

Generalizing the aforementioned minimum output entropy conjecture, the \emph{Constrained Minimum Output Entropy (CMOE)} conjecture states that the minimum output entropy of the above phase-covariant Gaussian channels over input states with a given entropy is achieved by Gaussian states~\cite{GuhaShapiroErkmen2007}. More precisely, the CMOE conjecture states that if $\Phi$ is a single-mode phase-covariant Gaussian channel, which by the above discussion is either attenuator, amplifier or additive-noise, then for any $m$-mode state $\rho$ we have
\begin{align}\label{eq:CMOE-introduction}
\frac 1m S\big( \Phi^{\otimes m}(\rho)\big)\geq S(\Phi(\tau)),
\end{align}
where $\tau$ is a single-mode thermal state satisfying $S(\tau) = \frac 1m S(\rho)$. 
This conjecture is in particular related to the problem of computing the capacity of quantum Gaussian broadcast channels~\cite{GuhaShapiroErkmen2007}. Using tools developed in~\cite{PalmaTrevisanGiovannetti2016Passive} the CMOE conjecture is proven in~\cite{DePalma+2017PRL, PalmaTrevisanGiovannettt2017Attenuator} \emph{in the special case of $m=1$.} It is also proven for arbitrary $m$ in the parameter regime that the channel $\Phi$ becomes entanglement breaking~\cite{DePalma2019EntBreaking}. 

Some other conjectures regarding the optimality of Gaussian states have also been proposed including an \emph{entropy photon-number inequality}~\cite{GuhaErkmenShapiro2008, DePalmaMariGiovannetti2014GeneralizationEPI} and a \emph{sharp Young's inequality for the beam-splitter}~\cite{DePalma+2018survey}.  For more details on \emph{quantum Gaussian optimizer conjectures} in quantum information theory we refer to~\cite{DePalma+2018survey}.

\subsection{Our contributions}

In this paper, we introduce a \emph{meta log-Sobolev inequality}  
for phase-covariant Gaussian quantum channels. In the Schr\"odinger picture, the Lindbladian of the semigroups of attenuator, amplifier and additive-noise channels mentioned above, takes the form
$$\cL(\rho) = \nu_0\Big(\frac 12 \{\bfa\bfa^\dagger, \rho\} - \bfa^\dagger \rho\bfa\Big)+ \nu_1\Big( \frac 12 \{\bfa^\dagger\bfa, \rho\} - \bfa \rho\bfa^\dagger\Big),$$
where $\nu_0, \nu_1\geq 0$ are non-negative constants and $\{X, Y\} = XY+YX$ denotes anti-commutator. Then, given any $\omega\geq 0$, we define the function 
\begin{align*}
	\Upsilon(\rho)= \frac{p}{p-1}\tr\Big( \cL\big(\rho^{1/p}\big) \rho^{1-1/p}\Big) + \omega\,\tr\big(\rho \bfa^\dagger \bfa\big) + S(\rho),
\end{align*}
over single-mode states $\rho$.  Later, we will see that for the choice of 
$$X=\sigma_\beta^{-\frac{1}{2p}} \rho^{\frac 1p} \sigma_\beta^{-\frac{1}{2p}},$$ 
and $p=2$, the first term in $\Upsilon(\rho)$ resembles the \emph{Dirichlet form} appearing on the right hand side of~\eqref{eq:Q-LSI}. In this case, $\Ent_{p, \sigma_\beta}(X)=D(\rho\| \sigma_\beta)$ can be written in terms of the other two terms in $\Upsilon(\rho)$. Thus, computing the infimum of $\Upsilon(\rho)$ 
is useful in obtaining the optimal log-Sobolev constant~$\alpha_p$. 

Our main result, which we call a meta log-Sobolev inequality, states that  optimal states in the minimization of $\Upsilon(\rho)$ are thermal, i.e.,
\begin{align}\label{eq:main-ineq-intro}
	\Upsilon(\rho) \geq \inf_{\tau: \text{ \rm thermal}} \Upsilon\big(\tau\big),
\end{align}
for any state $\rho$ of a single-mode bosonic system.

By using \eqref{eq:main-ineq-intro}, we turn the problem of computing the \emph{optimal} $p$-log-Sobolev constant $\alpha_p$ into an optimization problem over thermal states that are characterized by a single real parameter. We explicitly compute the optimal constant for the quantum Ornstein-Uhlenbeck semigroup for any $1\leq p\leq 2$:
\begin{align}\label{eq:intro-alpha-p}
	\alpha_p = \frac{p\hatp}{4\beta} e^{\beta/2}\big(1-e^{-\beta/p}\big)\big(1-e^{-\beta/\hatp}\big).
\end{align}
This result for $p=1$ recovers the modified log-Sobolev inequality of~\cite{CarlenMaas17}, and for $p=2$ is an improvement over the bound of~\cite{CarboneSasso2008}. 

We also study the quantum Ornstein-Uhlenbeck semigroup in the multimode case. It is well-known that log-Sobolev constants for classical Markov semigroups have the tensorization property (see, e.g.,~\cite{Mossel+2013}), meaning that considering the $m$-fold tensor product of the semigroup, the log-Sobolev constants of the resulting semigroup 
are the same as those of the original one. This property is not known to hold for arbitrary quantum Markov semigroups. Thus, we also consider the $m$-mode version of~\eqref{eq:Q-LSI}:
$$\widehat \alpha_{2}\Ent_{2, \sigma_{\beta}^{\otimes m}} (X)\leq \sum_{j=1}^m \tr\Big(\!\sqrt{\sigma_\beta^{\otimes m}} [\bfa_j, X]^\dagger \sqrt{\sigma_\beta^{\otimes m}} [\bfa_j, X]\Big),$$
where $X>0$ runs over $m$-mode operators, and $\widehat \alpha_{2}$ is the $2$-log-Sobolev constant for the $m$-fold tensor product semigroup. We show that for any $m>1$,
\begin{align}\label{eq:intro-alpha-2-m}
	\widehat \alpha_{2} \geq  \bigg(\frac{2+ \log(2m+1)}{\sinh(\beta/2)} + \frac{1}{\alpha_{2}}\bigg)^{-1}\!,
\end{align}
where $\alpha_2$ is given by~\eqref{eq:intro-alpha-p}.

To prove our result~\eqref{eq:intro-alpha-2-m}, we follow the approach of~\cite{CarboneSasso2008}. We use properties of the spectrum of the generator $\cL$ of the quantum Ornstein-Uhlenbeck semigroup to reduce the problem for arbitrary $m$-mode operators $X$ to operators that are diagonal in the number basis. To this end, we prove an entropic inequality which might be of independent interest. Next, restricting to diagonal operators, we obtain an essentially classical Markov semigroup that is known to satisfy the tensorization property. Thus, the problem for $m$-mode states is reduced to that of single-mode states for which we have already established the optimal log-Sobolev constant in~\eqref{eq:intro-alpha-p}. We note that, although, our first step of the proof is inspired by~\cite{CarboneSasso2008}, our penalty term that compensates for the reduction to diagonal states is smaller than that of~\cite{CarboneSasso2008}. Moreover, as mentioned above we compute the exact value of the $2$-log-Sobolev constant for diagonal states while~\cite{CarboneSasso2008} derives only a lower bound. Thus, our estimate on the $2$-log-Sobolev inequality is not just an $m$-mode generalization of~\cite{CarboneSasso2008}, but a twofold improvement thereof.

Our meta log-Sobolev inequality~\eqref{eq:main-ineq-intro} provides a general framework that can be applied in other settings as well. In this paper, using~\eqref{eq:main-ineq-intro} we also give an alternative proof of the CMOE conjecture~\eqref{eq:CMOE-introduction} in the case of $m=1$.

\paragraph{Outline of the paper:} The rest of this paper is organized as follows. In Section~\ref{sec:phase-cov-channels}, we review the class of single-mode phase-covariant Gaussian channels. We argue that these channels form semigroups and compute their Lindbladians. Section~\ref{sec:meta-LSI} is devoted to our meta log-Sobolev inequality~\eqref{eq:main-ineq-intro} and its proof. The log-Sobolev constants for the quantum Ornstein-Uhlenbeck semigroup are computed in Section~\ref{sec:LSI-OU}. Also, our proof of the CMOE conjecture is given in Section~\ref{sec:CMOE}. Final remarks are discussed in Section~\ref{sec:con} and some detailed computations are left for the appendices.

%************************************************************************
\section{Phase-covariance Gaussian channels}\label{sec:phase-cov-channels}

A single-mode bosonic continuous-variable system is described by the \emph{annihilation operator} $\bfa$ and it hermitian conjugate, the \emph{creation operator} $\bfa^\dagger$. These operators satisfy the bosonic commutation relation 
$$[\bfa, \bfa^\dagger]=\bfa\bfa^\dagger - \bfa^\dagger \bfa=1.$$ 
Fock states $\{\ket n:\, n\geq 0\}$, also called number states, form an orthonormal basis for the corresponding Hilbert space and are eigenvectors of $\bfa^\dagger \bfa$ called the \emph{number operator}: $\bfa^\dagger\bfa\ket{n}=n\ket{n}$. We indeed have $\bfa \ket n = \sqrt{n}\ket{n-1}$ and $\bfa^\dagger \ket n = \sqrt{n+1}\ket{n+1}$.
Single-mode quantum states $\rho$ are density operators ($\rho\succeq 0$ and $\tr(\rho)=1$) acting on this Hilbert space.

All states that we consider in this paper are assumed to be physical and hence have finite mean photon number. More precisely, for a pure state $\rho=\ket\psi\bra\psi$ we assume that $\ket\psi$ belongs to the domain of $\bfa$ and $\bra \psi \bfa^\dagger \bfa \ket\psi = \|\bfa\ket\psi\|^2 < +\infty$,  and more generally for a mixed state with eigen-decomposition $\rho = \sum_{j=0}^\infty  \lambda_j \ketbra{\psi_j}{\psi_j}$ we assume that $\ket{\psi_j}$, for any $j$ with $\lambda_j>0$, belongs to the domain of $\bfa$ and $\sum_j \lambda_j  \| \bfa \ket{\psi_j} \|^2 < +\infty$. The latter expression is indeed equal to $\tr(\bfa \rho \bfa^\dagger)$, yet we often write it as $\tr(\rho\bfa^\dagger \bfa)$, the conventional notation for the mean photon number, considering the following formal definition for the expectation value of an operator in quantum mechanics.\footnote{Note that although $\tr(\bfa \rho \bfa^\dagger)$ may be well-defined and finite, $\rho\bfa^\dagger \bfa$ may not be a trace class operator. For more details on the expectation value of unbounded operators in quantum mechanics we refer to~\cite{Busch2016QuantumMeasurement}.} For any positive semidefinite operator $X$, we formally define  $\tr(\rho X) = \sum_j \lambda_j \|X^{1/2} \ket{\psi_j}\|^2$ if $\ket{\psi_j}$ belongs to the domain of $X^{1/2}$ for all $j$ with $\lambda_j>0$, and let $\tr(\rho X) = +\infty$ otherwise.\footnote{Note that $\|\bfa\ket\psi\|=\| |\bfa^\dagger\bfa|^{1/2}\ket\psi  \|$, so this definition matches the one for the number operator.} Also, for a self-adjoint operator $X$ with decomposition $X=X_+-X_-$ where both $X_+, X_-$ are positive semidefinite, we define $\tr(\rho X) = \tr(\rho X_+) - \tr(\rho X_-)$.

Quantum states can be represented in terms of \emph{displacement (Weyl) operators} $D_\xi=\exp(\xi\bfa^\dagger-\bar\xi\bfa)$ with $\xi\in\mathbb{C}$ as
\begin{equation}\label{eq:state-rep-charac}
	\rho=\frac{1}{\pi} \int_{\mathbb{C}} \chi_\rho(\xi) D_{-\xi}\,   \dd^{2}\xi,
\end{equation}
where $\chi_\rho(\xi)=\tr(\rho D_\xi)$ is called the \emph{characteristic function}~\cite{CahillGlauber1969}. The characteristic function is the Fourier transform of the \emph{Wigner function}~\cite{HILLERY1984}. 

A quantum state is called \emph{Gaussian} if its characteristic and equivalently Wigner functions are Gaussian. Therefore, Gaussian states can be simply described in terms of the first-order and second-order moments of their Wigner function.
\emph{Thermal states} are an important class of Gaussian states. A thermal state with parameter $\beta>0$ that is proportional to the inverse of temperature, is diagonal in the Fock basis and is given by 
\begin{equation}\label{eq:thermal-state}
	\sigma_\beta=\frac{1}{\tr(e^{-\beta \bfa^\dagger\bfa})}e^{-\beta \bfa^\dagger\bfa}=(1-e^{-\beta})\sum_{n=0}^{\infty} e^{-n\beta}\ket{n}\bra{n}.
\end{equation} 
The characteristic function of this thermal state equals
\begin{align}\label{eq:chi-thermal}
\chi_{\sigma_\beta}(\xi) = e^{-\frac 1 2 \coth(\beta/2)|\xi|^2}.
\end{align}
At zero temperature ($\beta=\infty$) we obtain the vacuum state $\sigma_\infty=\ketbra{0}{0}$.

The evolution of quantum systems, in general, is described by quantum channels $\Phi$ that are linear completely positive and trace-preserving superoperators. Quantum channels that transform Gaussian states into Gaussian states are known as \emph{Gaussian channels}~\cite{HolevoWerner2001Evaluating,Braunstein2005,Weedbrook2012,Holevo+2013Book}. Gaussian channels that are unitary are called \emph{Gaussian unitaries}. %Any Gaussian unitary can be decomposed into displacement operators, \emph{passive transformations} and \emph{squeezing transformations}; see~\cite{SerafiniBook} for more details. 

A single-mode Gaussian channels is called \emph{phase-covariant}\footnote{Phase-covariant channels are also known as gauge-covariant channels.}~\cite{GiovannettiHolevoGarcia-Patron2015,Giovannetti2015majorization} if it satisfies
\begin{equation}
	\Phi\big(U_\theta \rho\, U_{\theta}^\dagger\big)=U_\theta \Phi(\rho) U_{\theta}^\dagger,
\end{equation}
for any state $\rho$ and $\theta\in[0,2\pi)$, where $U_\theta=e^{i\theta\, \bfa^\dagger\bfa}$ is the \emph{phase-rotation unitary}. 
By the classification of single-mode Gaussian channels~\cite{Holevo+2013Book, Holevo2007Structure}, a single-mode phase-covariant Gaussian channel can be described in terms of its action on the characteristic function. For any such channel there are parameters $\gamma, \lambda\geq 0$ such that
\begin{equation}\label{eq:Phas-Cov-Def}
	\chi_{\Phi(\rho)}(\xi)=e^{-\frac 1 2 \gamma|\xi|^2}\chi_\rho\big(\!\sqrt{\lambda}\xi\big),
\end{equation}  
where the complete positivity condition implies $\gamma\geq|1-\lambda|$~\cite{HolevoWerner2001Evaluating}. The channel $\Phi$ is called \emph{quantum limited} if $\gamma=|1-\lambda|$. By using the characteristic function of thermal states \eqref{eq:chi-thermal} in \eqref{eq:Phas-Cov-Def}, one can verify that phase-covariant Gaussian channels transform thermal states to thermal states.

Single-mode phase-covariant Gaussian channels consist of three classes of \emph{attenuator channels} corresponding to $0\leq\lambda<1$, \emph{additive-noise channels} corresponding to $\lambda=1$, and \emph{amplifier channels} corresponding $1<\lambda$. These are important physical channels in describing dynamics of continuous-variable quantum systems.

A crucial property of single-mode phase-covariant Gaussian channels is that they admit semigroup structures~\cite{HeinosaariHolevoWolf2010,GiovannettiHolevoLloydMaccone2010}. Specifically, any such channel can be written as $\Phi_{t_0}=e^{-t_0\cL}$ for some $t_0\geq 0$ where $\big\{\Phi_{t} = e^{-t\cL}:\, t\geq 0\big\}$ is a semigroup of single-mode phase-covariant Gaussian channels. We show in Appendix~\ref{app:generators} that the corresponding Lindbladian $\cL$ of such a semigroup takes the form
\begin{equation}\label{eq:LG-general form}
\cL(\rho) = \nu_0\Big(\frac 12 \{\bfa\bfa^\dagger, \rho\} - \bfa^\dagger \rho\,\bfa\Big)+ \nu_1\Big( \frac 12 \{\bfa^\dagger\bfa, \rho\} - \bfa\, \rho\,\bfa^\dagger\Big),
\end{equation}
where $\{X, Y\} = XY+YX$ denotes the anti-commutator, and $\nu_0, \nu_1\geq 0$ are parameters determining the semigroup. We argue in Appendix~\ref{app:generators} that three ranges for the parameters $\nu_0, \nu_1$ give the three classes of single-mode phase-covariant channels: $\nu_1>\nu_0$ corresponds to 
attenuator channels, $\nu_0=\nu_1$ corresponds to  additive-noise channels, and $\nu_1<\nu_0$ corresponds to 
amplifier channels.

Attenuator channels, in general, can be physically modeled by applying a \emph{beam splitter unitary} with the transmissivity of $0\leq \lambda\leq1$ on the system and an auxiliary system in the thermal state $\sigma_\beta$, and then tracing out the second subsystem: 
\begin{equation}\label{eq:U-att}
	\Phiatt_\lambda(\rho)=\tr_2\Big( U_{\text{BS},\lambda} (\rho\otimes \sigma_\beta) U_{\text{BS},\lambda}^\dagger \Big).
\end{equation}
The characteristic function of the output state is given by
\begin{equation*}
	\chi_{\Phiatt_\lambda(\rho)} (\xi)=e^{-\frac1 2 \coth(\beta/2)(1-\lambda)|\xi|^2}\chi_\rho(\sqrt{\lambda}\xi).
\end{equation*}
Choosing the transmissivity parameter $\lambda_t=e^{-2ct}$ as a function of time, where $c>0$ is some constant, we obtain a semigroup of attenuation channels. The generator of this semigroup given in \eqref{eq:LG-general form} has parameters $\nu_0 =  c \left(\coth(\beta/2)-1\right)$ and $\nu_1 =c \left(\coth(\beta/2)+1\right)$. Note that for the special choice of $c=\sinh(\beta/2)$ we have $\nu_0=e^{-\beta/2}$ and $\nu_1=e^{\beta/2}$. This semigroup of  
attenuator channels is sometimes called the \emph{quantum Ornstein-Uhlenbeck semigroup}.

Amplifier channels can be described by replacing the beam splitter unitary in the above model by a \emph{two-mode squeezing unitary} to get
\begin{equation}\label{eq:U-amp}
	\Phiamp_\lambda(\rho)=\tr_2\Big( U_{\text{2S},\lambda} (\rho\otimes \sigma_\beta) U_{\text{2S},\lambda}^\dagger \Big),
\end{equation}
where $\lambda\geq1$ is the squeezing parameter.
In this case, the relation between input and output characteristic functions becomes
\begin{equation*}
	\chi_{\Phiatt_\lambda(\rho)} (\xi)=e^{-\frac1 2 \coth(\beta/2)(\lambda-1)|\xi|^2}\chi_\rho(\sqrt{\lambda}\xi).
\end{equation*}
Choosing $\lambda_t=e^{2ct}$ with $c>0$ to be a function of time, we obtain a semigroup with the generator corresponding to parameters $\nu_0=c\left(\coth(\beta/2)+1\right)$ and $\nu_1=c\left(\coth(\beta/2)-1\right)$ in \eqref{eq:LG-general form}.  

Additive-noise channels can be modeled by applying a displacement operator whose parameter is chosen at random according to a Gaussian probability distribution:
\begin{equation*}
	\Phiad_\gamma(\rho)=\frac{2}{\pi\gamma}\int e^{-\frac{2}{\gamma}|\xi|^2} D_\xi\rho D^{\dagger}_\xi\, \dd^2\xi.
\end{equation*} 
This channel in terms of characteristic functions can be written as
\begin{equation*}
	\chi_{\Phiad_\gamma(\rho)} (\xi)=e^{-\frac 12\gamma|\xi|^2}\chi_\rho(\xi).
\end{equation*}
Again, by setting $\gamma_t=2ct$ with $c>0$ we obtain a semigroup whose generator given by~\eqref{eq:LG-general form} has parameters $\nu_0=\nu_1=c$. This semigroup is sometimes called the \emph{quantum heat semigroup}.

We emphasize that by the above discussion any single-mode phase-covariant channel~\eqref{eq:Phas-Cov-Def} can be viewed as a member of one of the above three semigroups. The point is that, the above choices of parameters $\lambda_t, \gamma_t$ cover all the valid ranges of $\lambda, \gamma$ in~\eqref{eq:Phas-Cov-Def}.

In this paper, we also consider $m$-mode bosonic systems, described by $m$ pairs of annihilation and creation operators $\{\bfa_1, \bfa_1^\dagger, \dots, \bfa_m, \bfa_m^\dagger \}$ satisfying $[\bfa_{i}, \bfa_j]=0$ and $[\bfa_{i}, \bfa_j^\dagger]=\delta_{i,j}$. Vectors of the associated tensor product Hilbert space can be expressed in terms of $m$-mode number states $\{\ket{n_1,\dots,n_m}:\, n_1,\dots, n_m\geq 0\}$.  Also, the characteristic function for an $m$-mode state $\rho$ is given by $\chi_\rho(\xi)=\tr\big(\rho D_{\xi}\big)$, where $\xi=(\xi_1,\dots,\xi_m)\in \mathbb C^m$ and $D_\xi = D_{\xi_1}\otimes \dots \otimes D_{\xi_m}$. As in the single-mode case, physical multimode states $\rho$ considered in this paper have finite mean photon number: $\tr(\rho H_m)<+\infty$, where $H_m=\sum_{j=1}^m \bfa_j^\dagger \bfa_j$ is the $m$-mode number operator. The definitions of Gaussian channels and unitaries can be extended to the multimode case. Any Gaussian unitary can be decomposed into displacement operators, mutlimode \emph{passive transformations}, which preserve the mean photon number, and single-mode \emph{squeezing transformations}; see~\cite{SerafiniBook} for more details.

%************************************************************************
\section{Meta log-Sobolev inequality}\label{sec:meta-LSI}

As discussed in the previous section, all the Lindbladians associated to single-mode phase-covariant channels have the form $\cL=\nu_0\cL_{0} + \nu_1 \cL_1$ for some $\nu_0, \nu_1\geq 0$ where
\begin{align}\label{eq:L0-L1}
\cL_0(X) = \frac 12 \{\bfa\bfa^\dagger, X\} - \bfa^\dagger X\bfa, \qquad \cL_1(X) = \frac 12 \{\bfa^\dagger\bfa, X\} - \bfa X\bfa^\dagger.
\end{align}
Let $\langle \cdot, \cdot \rangle$ denote the \emph{Hilbert-Schmidt inner product:} 
$$\langle  X, Y\rangle= \tr(X^\dagger Y).$$ 
Also, let $\cL^*$ be the adjoint of the generator $\cL$ with respect to this inner product: $\langle X, \cL(Y)\rangle = \langle \cL^*(X), Y\rangle$.  Indeed, $\cL^*$ is the generator in the \emph{Heisenberg picture} given by
\begin{align*}
\cL^*(X) = \nu_0\Big( \frac{1}{2}\{\bfa \bfa^\dagger, X\}-\bfa X\bfa^\dagger \Big)+\nu_1\Big( \frac{1}{2}\{\bfa^\dagger \bfa, X\} -\bfa^\dagger X\bfa\Big).
\end{align*}
 
Let $p\geq 1$ and $\hatp$ be the \emph{H\"older conjugate} of $p$ given by\footnote{If $p=1$, we let $\hatp=+\infty$.} 
$$\frac{1}{p} + \frac{1}{\hatp}=1.$$ 
Then, for any single-mode quantum state $\rho$ define
\begin{align}\label{eq:def-Upsilon}
\Upsilon(\rho)&:=\hatp  \big\langle \cL\big(\rho^{1/p}\big), \rho^{1/\hatp}\big\rangle + \omega \tr(\rho \bfa^\dagger \bfa) + S(\rho)\nonumber\\
&\, =\hatp   \Big(\!\nu_0\big\langle \cL_{0}\big(\rho^{1/p}\big), \rho^{1/\hatp}\big\rangle+\nu_1\big\langle \cL_{1}\big(\rho^{1/p}\big), \rho^{1/\hatp}\big\rangle\!\Big) + \omega\tr(\rho \bfa^\dagger \bfa) +S(\rho),
\end{align}
where $\omega\geq 0$ is a fixed parameter and $S(\rho)= -\tr(\rho\log \rho)$ is the von Neumann entropy of the state. The factor $\hatp$ in the first term of $\Upsilon(\rho)$ is for the sake of normalization in the limiting case of $p\to 1^+$. Indeed, since $\cL^*(I)=0$, we have 
$$\lim_{p\to 1^+} \hatp \big\langle \cL\big(\rho^{1/p}\big), \rho^{1/\hatp}\big\rangle =\lim_{p\to 1^+} \hatp\,\big \langle \rho^{1/p}, \cL^*\big(\rho^{1/\hatp}\big)\big\rangle = \langle \rho, \cL^*(\log \rho)\rangle= \langle \cL(\rho), \log \rho\rangle. $$
Thus, by convention for $p=1$ we let $\hatp\, \big\langle \cL\big(\rho^{1/p}\big), \rho^{1/\hatp}\big\rangle = \langle \cL(\rho), \log \rho\rangle$.

A straightforward computation verifies that
\begin{align}\label{eq:Upsilon-equiv-exp}
\Upsilon(\rho) =\hatp  \bigg(\!\nu_0 \Big[\tr\big(\rho\bfa\bfa^\dagger\big) -  \tr\big(  \rho^{1/p} \bfa \rho^{1/\hatp} \bfa^\dagger\big) \Big] +\nu_1\Big[  \tr\big(\rho\bfa^\dagger\bfa\big)-  \tr\big(  \rho^{1/p} \bfa^\dagger \rho^{1/\hatp} \bfa \big)\!\Big]\!\bigg) + \omega \tr\big(\rho \bfa^\dagger\bfa\big) + S(\rho).
\end{align}
We may think of this equation as the starting definition of $\Upsilon(\rho)$. In Appendix~\ref{app:range-Upsilon}, we show that if $\rho$ has a finite mean photon number, then all the terms in the above equation are finite and $\Upsilon(\rho)$ given by~\eqref{eq:Upsilon-equiv-exp} is well-defined.

One of the main technical contributions of our work is that the infimum of $\Upsilon(\rho)$ over states $\rho$ is achieved at single-mode thermal states. To this end, it would be beneficial to compute $\Upsilon(\rho)$ for thermal states. 
Using~\eqref{eq:Upsilon-equiv-exp}, for a single-mode thermal state,\footnote{Hereafter, we use $\tau$ to denote a thermal state on which we optimize, and save $\sigma=\sigma_\beta$ for the \emph{reference} thermal state in our description of phase-covariant Gaussian channels; see equations~\eqref{eq:U-att} and~\eqref{eq:U-amp}.} 
$$\tau=\tau_x=(1-x)\sum_n x^n\ketbra{n}{n},$$
where we put $0<x=e^{-\beta}<1$ in~\eqref{eq:thermal-state}, we have 
\begin{align}\label{eq:Upsilon-for-thermal}
\Upsilon(\tau) &= \hatp (1-x)  \sum_n \Big( \nu_0 (1-x^{1/\hatp}) (n+1)   x^n + \nu_1 (1-x^{-1/\hatp}) n x^n \Big)\nonumber\\
&\quad ~+ \omega(1-x) \sum_n nx^n  - (1-x)\sum_n x^n\log \big((1-x)x^n\big)\nonumber\\
& =   \hatp (1-x)\Big(\nu_0 (1-x^{1/\hatp}) \frac{1}{(1-x)^2} + \nu_1 (1-x^{-1/\hatp}) \frac{x}{(1-x)^2}  \Big)\nonumber\\
&\quad~ + \omega(1-x) \frac{x}{(1-x)^2} - (1-x) \frac{x}{(1-x)^2}  \log x- \log(1-x) \nonumber\\
& =  \frac{\hatp}{1-x} \Big(\nu_0 (1-x^{1/\hatp})+ \nu_1(x- x^{1/p})   \Big) + \omega\frac{x}{1-x} -  \frac{x}{1-x}  \log x- \log(1-x). 
\end{align}
Optimizing over the choice of the thermal state $\tau$, we define 
\begin{align}\label{eq:def-eta-thermal}
\eta_{\thermal}&:= \inf_{\tau: \text{ \rm thermal}} \Upsilon(\tau)\nonumber\\
&\,=\inf_{0<x<1}\, \frac{\hatp}{1-x} \Big(\nu_0 (1-x^{1/\hatp})+ \nu_1(x- x^{1/p})   \Big) + \omega\frac{x}{1-x} -  \frac{x}{1-x}  \log x- \log(1-x).
\end{align}

\medskip

\begin{theorem}\label{thm:main-technical-result} \emph{[Meta log-Sobolev inequality]}
For any $\nu_0, \nu_1, \omega\geq 0$ and $ p\geq 1$ define $\Upsilon(\rho)$ by~\eqref{eq:def-Upsilon}. Let $\eta_{\thermal}$ be the infimum of $\Upsilon(\tau)$ over thermal states as in~\eqref{eq:def-eta-thermal}. Then, for any single-mode quantum state $\rho$ with finite mean photon number we have 
\begin{align}\label{eq:meta-LSI}
\Upsilon(\rho)\geq \eta_\thermal.
\end{align}
\end{theorem}

\medskip
The proof of this theorem is broken into two steps:

\begin{itemize}
\item[(i)] The first step is to reduce the problem for arbitrary states $\rho$ to states that are diagonal in the Fock basis. To this end, we use ideas developed in~\cite{PalmaTrevisanGiovannetti2016Passive}. We show that, fixing the eigenvalues of $\rho$ and rotating its eigen-basis, we can obtain a diagonal state $\widehat \rho$ that satisfies $\Upsilon(\rho)\geq \Upsilon(\widehat \rho)$. We then conclude that diagonal states are sufficient when minimizing $\Upsilon(\cdot)$.

\item[(ii)] In the second step, we show that the optimal diagonal states are thermal. The proof idea in this step is extracted from \emph{tensorization type} arguments. Tensorization was first used by Gross~\cite{Gross75} for proving his celebrated log-Sobolev inequality. The idea in~\cite{Gross75} is that by the central limit theorem, an expectation value with respect to a  Gaussian distribution can be understood as the limit $n\to +\infty$ of the expectation value of some lifted function on the product space $\{0,1\}^n$ with respect to a product probability measure. Here, being interested in thermal states, our distributions of reference are \emph{geometric distributions}, and geometric distributions can be understood in terms of the first success in a sequence of independent Bernoulli trials. Thus, it is natural to employ a tensorization type argument in order to reduce the optimization over general single-mode diagonal states to thermal ones. In our proof, we do not directly refer to Bernoulli trials, yet our intuition on why and how it works is really rooted in a tensorization argument as described here.

\end{itemize}

\begin{proof}
(i) In the first step of the proof we  show that for any state $\rho$, there is a state $\widehat \rho$ that  is diagonal in the  Fock  basis and satisfies $\Upsilon(\rho)\geq \Upsilon(\widehat  \rho)$.  

Let 
$$\rho = \sum_{n=0}^\infty  \lambda_n \ketbra{\psi_n}{\psi_n},$$
be the eigen-decomposition of $\rho$ where $\lambda_0 \geq \lambda_1\geq  \cdots \geq 0$ are the eigenvalues of $\rho$ and $\{\ket{\psi_n}:\, n\geq 0\}$ is its eigen-basis. In this case, there are $r_k=\lambda_k-\lambda_{k+1}\geq0$ such that
$$\rho = \sum_{k=0}^{\infty} r_k P_k,$$
where $P_k = \sum_{n=0}^{k} \ketbra{\psi_n}{\psi_n}$ is the projection operator on the first $k+1$ eigenvectors of $\rho$. Now  we define
$$\widehat  \rho = \sum_{k=0}^\infty  r_k \Pi_k,$$
where $\Pi_k = \sum_{n=0}^{k}  \ketbra{n}{n}$ is the projection operator on  the first  $k+1$ vectors in the Fock basis. 
We note  that by definition, $\widehat \rho = \sum_{k=0}^n \lambda_k \ketbra{n}{n}$ is diagonal in the Fock basis and shares the same eigenvalues with $\rho$. Therefore, they have the same entropy 
\begin{align}\label{eq:S-rho-hatrho}
S(\widehat \rho)= S(\rho).
\end{align}
Next, we note that the infimum of $\tr(Q_k \bfa^\dagger \bfa)$ over  
projectors $Q_k$ with rank $k+1$, equals the sum of the first $k+1$ smallest eigenvalues of $\bfa^\dagger \bfa$, and is achieved by $Q_k = \Pi_k$. 
This implies that the mean photon number of $\rho$ is lower bounded by that of $\widehat \rho $:
\begin{align}\label{eq:adaggera-rho-hatrho}
\tr(\rho\bfa^\dagger \bfa) & = \sum_k r_k \tr(P_k \bfa^\dagger \bfa) \geq \sum_k r_k \tr(\Pi_k \bfa^\dagger \bfa) = \tr(\widehat \rho \bfa^\dagger \bfa). 
\end{align}
We also have 
$$\big\langle \cL_{0}\big(\rho^{1/p}\big), \rho^{1/\hatp}\big\rangle = \tr(\rho \bfa\bfa^\dagger) - \tr\big(\rho^{1/p} \bfa \rho^{1/\hatp} \bfa^\dagger\big) = \tr\big(\rho^{1/p}\rho^{1/\hatp} \bfa\bfa^\dagger\big) - \tr\big(\rho^{1/p} \bfa \rho^{1/\hatp} \bfa^\dagger\big).$$
We note that $\rho^{1/p} = \sum_n \lambda_n^{1/p} \ketbra{\psi_n}{\psi_n}$ with $\lambda_0^{1/p}\geq \lambda_1^{1/p}\geq \cdots \geq 0$. Then, there are $r_{p,k} \geq 0$ such that $\rho^{1/p} = \sum_k  r_{p, k} P_k$. Similarly,  there are  $r_{\hatp, k}\geq 0$ such that $\rho^{1/\hatp} = \sum_k  r_{\hatp, k} P_k$. We also have $\widehat \rho^{1/p} = \sum_k  r_{p, k} \Pi_k$ and $\widehat \rho^{1/\hatp} = \sum_k  r_{\hatp, k} \Pi_k$. Therefore, we have 
\begin{align*}
\big\langle \cL_0 &\big(\rho^{1/p}\big), \rho^{1/\hatp}\big\rangle  = \sum_{k, \ell} r_{p,k}\,r_{\hat p, \ell} \Big(\! \tr\big(P_kP_\ell  \bfa\bfa^\dagger\big) - \tr\big( P_k \bfa P_\ell \bfa^\dagger \big)\!\Big)\\
& = \sum_{k\geq  \ell} r_{p,k}\,r_{\hat p, \ell} \Big(\! \tr\big(P_kP_\ell \bfa \bfa^\dagger \big) - \tr\big( P_k \bfa P_\ell \bfa^\dagger \big)\!\Big) + \sum_{k<  \ell} r_{p,k}\,r_{\hat p, \ell} \Big(\! \tr\big(P_kP_\ell \bfa \bfa^\dagger \big) - \tr\big( P_k \bfa P_\ell \bfa^\dagger \big)\!\Big)\\
& = \sum_{k\geq  \ell} r_{p,k}\,r_{\hat p, \ell} \Big(\! \tr\big(P_\ell \bfa\bfa^\dagger \big) - \tr\big( P_k \bfa P_\ell \bfa^\dagger \big)\!\Big) + \sum_{k<  \ell} r_{p,k}\,r_{\hat p, \ell} \Big(\! \tr\big(P_k \bfa \bfa^\dagger\big) - \tr\big( P_k \bfa P_\ell \bfa^\dagger \big)\!\Big)\\
& \geq \sum_{k\geq  \ell} r_{p,k}\,r_{\hat p, \ell} \Big(\! \tr\big(P_\ell \bfa \bfa^\dagger\big) - \tr\big(  \bfa P_\ell\bfa^\dagger \big)\!\Big) + \sum_{k<  \ell} r_{p,k}\,r_{\hat p, \ell} \Big(\! \tr\big(P_k \bfa \bfa^\dagger\big) - \tr\big( P_k \bfa \bfa^\dagger \big)\!\Big),
\end{align*}
where in the last inequality we use $P_k, P_\ell\preceq I$. Thus, using the commutation relation $[\bfa, \bfa^\dagger]= 1$, we get
\begin{align*}
\big\langle \cL_0 \big(\rho^{1/p}\big), \rho^{1/\hatp}\big\rangle  
 \geq  \sum_{k\geq  \ell} r_{p,k}\,r_{\hat p, \ell}\,  \tr(P_k )  =   \sum_{k\geq  \ell} r_{p,k}\,r_{\hat p, \ell} (k+1).
\end{align*}
Repeating the same computation for $\widehat  \rho$, we observe that the support  of $\bfa \Pi_\ell \bfa^\dagger$ is included in the support of $\Pi_k$ if $k\geq \ell$. This implies $\tr\big( \Pi_k \bfa \Pi_\ell \bfa^\dagger \big) = \tr\big( \bfa \Pi_\ell \bfa^\dagger \big)$. Similarly, if $k<\ell$, then the support  of $\bfa^\dagger \Pi_k\bfa$ is inside the support of $\Pi_\ell$, which implies $\tr\big( \Pi_k \bfa \Pi_\ell \bfa^\dagger \big)=\tr\big( \bfa^\dagger \Pi_k \bfa \Pi_\ell  \big) = \tr\big( \Pi_k \bfa  \bfa^\dagger \big)$. Therefore, we have
\begin{align}\label{eq:L-0-rho-hatrho}
\big\langle \cL_0 \big(\rho^{1/p}\big), \rho^{1/\hatp}\big\rangle  
 \geq     \sum_{k\geq  \ell} r_{p,k}\,r_{\hat p, \ell} (k+1) = \big\langle \cL_0 \big(\widehat \rho^{1/p}\big), \widehat \rho^{1/\hatp}\big\rangle.
\end{align}
By using the same argument, we also have
\begin{align}\label{eq:L-1-rho-hatrho}
\big\langle \cL_1 \big(\rho^{1/p}\big), \rho^{1/\hatp}\big\rangle  
 \geq   \big\langle \cL_1 \big(\widehat \rho^{1/p}\big), \widehat \rho^{1/\hatp}\big\rangle.
\end{align}
Putting~\eqref{eq:S-rho-hatrho}-\eqref{eq:L-1-rho-hatrho} together and using the fact that $\nu_0, \nu_1, \omega\geq 0$ we conclude that $\Upsilon(\rho)\geq \Upsilon(\widehat \rho)$ for $p>1$. The same inequality for $p=1$ is obtained by taking the limit $p\to 1^+$. Therefore, to prove $\Upsilon(\rho)\geq \eta_{\thermal}$ in Theorem~\ref{thm:main-technical-result}, it suffices to restrict to diagonal states in the Fock basis.

\medskip
\medskip
\noindent
(ii) We now, in the second step, show that the optimal diagonal states that minimize $\Upsilon(\rho)$ are thermal. Let
$$\rho= \sum_n \lambda_n \ketbra{n}{n},$$
be the eigen-decomposition of $\rho$.
A straightforward computation yields 
\begin{align}\label{eq:Upsilon-single-mode-diagonal}
\Upsilon(\rho) & = \hatp \sum_n \Big( \nu_0 (n+1)\Big(\lambda_n-\lambda_n^{1/p}\lambda_{n+1}^{1/\hat p}\Big) + \nu_1 n \Big(\lambda_n - \lambda_n^{1/p}\lambda_{n-1}^{1/\hat p} \Big)\!\Big)\nonumber\\
&\quad ~ + \omega \sum_n n\,\lambda_n  -\sum_n \lambda_n\log \lambda_n.
\end{align}

By the definition of $\eta_\thermal$ in~\eqref{eq:def-eta-thermal},  for any $0\leq x\leq 1$ we have\footnote{Here, we take the limits $x\to 0^+$ and $x\to 1^-$ to include the boundary values.}
$$\eta_{\thermal} (1-x)\leq\hatp \Big(\nu_0 (1-x^{1/\hatp})+ \nu_1 (x-x^{1/p})   \Big) + \omega x -  x  \log x- (1-x)\log(1-x).$$
For any $\ell\geq 0$ let 
$$s_\ell = \sum_{n=\ell}^\infty \lambda_n, \qquad x_\ell = \frac{s_{\ell+1}}{s_\ell}.$$  
We note that $s_\ell = \lambda_\ell + s_{\ell+1}$ and $0\leq x_\ell\leq 1$. Therefore,  
\begin{align*}
 \eta_{\thermal} (1-x_\ell)
\leq \hatp \Big(\nu_0 \Big(1-x_\ell^{1/\hatp}\Big)+ \nu_1 \Big(x-x_\ell^{1/p}\Big) \!  \Big) + \omega x_\ell -  x_\ell  \log x_\ell- (1-x_\ell)\log(1-x_\ell).
\end{align*}
Multiplying both sides by $s_\ell$ and summing over $\ell$, we obtain
\begin{align}\label{eq:tensorization-ineq}
 \eta_{\thermal} \sum_{\ell=0}^\infty s_\ell (1-x_\ell)
&\leq \hatp \sum_{\ell=0}^\infty s_\ell \Big(\nu_0 \Big(1-x_\ell^{1/\hatp}\Big)+ \nu_1\Big(x_\ell-x_\ell^{1/p}\Big) \!  \Big) \nonumber\\
&\quad ~ + \omega \sum_{\ell=0}^\infty s_\ell x_\ell -  \sum_{\ell=0}^\infty s_\ell\big(x_\ell  \log x_\ell+ (1-x_\ell)\log(1-x_\ell)\big).
\end{align}
We compute and estimate each term in the above equation as follows.
First, we have 
$$\sum_{\ell=0}^\infty s_\ell (1-x_\ell) =  \sum_{\ell=0}^\infty \lambda_\ell = 1.$$
Second, by applying H\"older's inequality, we obtain
\begin{align}\label{eq:applying-Holder}
\sum_{\ell=0}^\infty  s_\ell (1-x_\ell^{1/\hat p})  & = \sum_{\ell=0}^\infty s_\ell -\sum_{\ell=0}^\infty    s_\ell^{1/p} s_{\ell+1}^{1/\hat p} \nonumber \\
& = \sum_{\ell=0}^\infty s_\ell -\sum_{\ell=0}^\infty  \Big(\sum_{n=\ell}^\infty \lambda_n\Big)^{\!1/p}  \Big(\sum_{n=\ell}^\infty \lambda_{n+1}\Big)^{\!1/\hatp} \nonumber \\
&\leq \sum_{\ell=0}^\infty  \sum_{n=\ell}^\infty \lambda_n - \sum_{\ell=0}^\infty  \sum_{n=\ell}^\infty \lambda_n^{1/p} \lambda_{n+1}^{1/\hatp} \nonumber \\
& = \sum_{n=0}^\infty (n+1)\Big(\lambda_n - \lambda_n^{1/p} \lambda_{n+1}^{1/\hatp}\Big).
\end{align}
Similarly, we obtain
\begin{align*}
\sum_{\ell=0}^\infty  s_\ell (x_\ell -x_\ell^{1/p}) \leq \sum_{n=0}^\infty n \big(\lambda_n- \lambda_n^{1/p} \lambda_{n-1}^{1/\hatp}\big).
\end{align*}
We also have 
$$\sum_{\ell=0}^\infty s_\ell x_\ell = \sum_{\ell=0}^\infty s_{\ell+1}  =\sum_{\ell=0}^\infty \sum_{n=\ell+1}^\infty \lambda_n = \sum_{n=0}^\infty n\lambda_n.$$
Finally, the last term of \eqref{eq:tensorization-ineq} becomes
\begin{align*}
\sum_{\ell=0}^\infty s_\ell\big(x_\ell  \log x_\ell+ (1-x_\ell)\log(1-x_\ell)\big) 
& =  \sum_{\ell=0}^\infty s_\ell\Big(\frac{s_{\ell+1}}{s_\ell}  \log \frac{s_{\ell+1}}{s_\ell}+ \frac{\lambda_{\ell}}{s_\ell}\log\frac{\lambda_{\ell}}{s_\ell}\Big)\\
& =  \sum_{\ell=0}^\infty \Big(s_{\ell+1}  \log s_{\ell+1}- s_{\ell+1}\log s_\ell+ \lambda_\ell \log \lambda_\ell -\lambda_{\ell}\log s_\ell\Big)\\
& =  \sum_{\ell=0}^\infty \Big(s_{\ell+1}  \log s_{\ell+1}- s_{\ell}\log s_\ell+ \lambda_\ell \log \lambda_\ell \Big)\\
& = \sum_{\ell=0}^\infty \lambda_\ell\log\lambda_\ell + \lim_{k\to+\infty}  \sum_{\ell=0}^{k-1} \Big(s_{\ell+1}  \log s_{\ell+1}- s_{\ell}\log s_\ell\Big)\\
& = \sum_{\ell=0}^\infty \lambda_\ell\log\lambda_\ell - s_0\log s_0 + \lim_{k\to+\infty}  s_{k} \log s_{k}\\
& = \sum_{\ell=0}^\infty \lambda_\ell\log\lambda_\ell,
\end{align*}
where in the last line we use $s_0=\sum_n \lambda_n=1$ and $\lim_{k\to +\infty} s_k=0$. Using these equations in~\eqref{eq:tensorization-ineq} and comparing to~\eqref{eq:Upsilon-single-mode-diagonal} we arrive at $\Upsilon(\rho)\geq \eta_\thermal$. 

The proof for $p=1$ is similar; we only need to replace the H\"older inequality in~\eqref{eq:applying-Holder} with 
$$-s_\ell \log x_{\ell}\leq  \sum_{n=\ell}^\infty \lambda_n\big(\log \lambda_{n}-\log \lambda_{n+1}\big),$$
that is derived from H\"older's inequality by taking an appropriate limit.\footnote{This inequality can also be derived from the non-negativity of the Kullback-Leibler divergence.} 
\end{proof}

We remark that in part (ii) of the above proof, when the diagonal state $\rho$ is thermal, the parameter $x_\ell$ is independent of $\ell$. Moreover, in this case, the H\"older inequality applied in~\eqref{eq:applying-Holder} is tight. 

Note that our meta log-Sobolev inequality can be viewed as a generalization of a log-sobolev inequality introduced in~\cite[Theorem~3.1]{DTG18} which holds only for the generator of  the quantum-limited attenuator channel.

%******************************
\subsection{Meta log-Sobolev inequality for multimode states}

For some applications it is useful to consider the function $\Upsilon(\cdot)$ in the \emph{multimode} case and prove a generalization of Theorem~\ref{thm:main-technical-result} for multimode quantum states. In this subsection, we establish such a generalization but for the special cases of multimode states that can be prepared by applying a Gaussian unitary on  multimode states that are diagonal in the Fock basis, which include all multimode Gaussian states, and multimode states that can be prepared by applying a passive Gaussian unitary on any product state.

For any $m$-mode quantum state $\rho$ with finite mean photon number define 
\begin{align}\label{eq:def-Upsilon-m}
\Upsilon_{m}(\rho)&:=\frac{\hatp}{m} \sum_{j=1}^m  \big\langle \cL_j\big(\rho^{1/p}\big), \rho^{1/\hatp}\big\rangle + \frac{\omega}{m}\sum_{j=1}^m\tr(\rho \bfa_j^\dagger \bfa_j) +\frac 1m S(\rho),
\end{align}
where as before $\nu_{0}, \nu_{1},\omega\geq 0$ and $\cL_j = \nu_{0}\cL_{0, j}+\nu_{1}\cL_{1, j}$ is the Lindbladian acting on the $j$-th mode. Note that we still use $\Upsilon(\rho)=\Upsilon_{1}(\rho)$ for  the single-mode case. To establish our results on the $m$-mode generalizations of Theorem~\ref{thm:main-technical-result} we first state a lemma.

\begin{lemma}\label{lem:Gaussian-unitaries}
For any $m$-mode state $\rho$ the followings hold:
\begin{enumerate}
\item[{\rm (i)}] Let $\rho' = D_{ \xi} \rho D_\xi^\dagger$ where $D_\xi = D_{\xi_1}\otimes \dots \otimes D_{\xi_m}$ is an $m$-mode displacement operator. Suppose that $\tr(\rho \bfa_j)=0$ for all $j$. Then, $\Upsilon_m(\rho')\geq \Upsilon_m(\rho)$.
\item[{\rm (ii)}] Let $\rho' = U\rho U^\dagger$ where $U$ is a passive transformation, i.e., $U$ is a Gaussian unitary that commutes with $H_m=\sum_j \bfa_j^\dagger \bfa_j$. Then, $\Upsilon_m(\rho')= \Upsilon_m(\rho)$. 
\item[{\rm (iii)}] Let $\rho' = S_r\rho S_r^\dagger$ where $S_r = e^{\frac 12\sum_j r_j(\bfa_j^2 - (\bfa_j^\dagger)^2)}$ is a squeezing transformation. Suppose that $\tr\big(\rho \bfa_j^2\big)=\tr\big( \rho^{1/p} \bfa_j \rho^{1/\hatp} \bfa_j\big)=0$ for all $j$.  Then, $\Upsilon_m(\rho')\geq \Upsilon_m(\rho)$. 
\end{enumerate}
\end{lemma}

\begin{proof} 
(i) We note that $D_\xi^\dagger \bfa_j D_\xi = \bfa_j + \xi_j$. Therefore, $D_\xi^\dagger \bfa_j^\dagger \bfa_j D_\xi = \bfa_j^\dagger \bfa_j + \xi_j \bfa_j^\dagger + \bar \xi_j \bfa_j + |\xi_j|^2$. Moreover, 
$$\tr(\rho' \bfa_j^\dagger \bfa_j) = \tr(\rho D_\xi^\dagger \bfa_j^\dagger \bfa_k D_\xi) = \tr(\rho \bfa_j^\dagger \bfa_j) + |\xi_j|^2,$$
where for the second equality we use the assumption $\tr(\rho \bfa_j)=0$. We similarly have 
$$\tr\big({\rho'}^{1/p} \bfa_j {\rho'}^{1/\hatp}\bfa_j^\dagger\big) = \tr\big({\rho}^{1/p} \bfa_j {\rho}^{1/\hatp}\bfa_j^\dagger \big) + |\xi_j|^2.$$
Therefore, we have  
$$\big\langle \cL_{0,j}(\rho'^{1/p}), \rho'^{1/\hatp}\big\rangle =\tr(\rho' \bfa_j\bfa_j^\dagger )-\tr\big({\rho'}^{1/p} \bfa_j {\rho'}^{1/\hatp}\bfa_j^\dagger\big) =\big\langle \cL_{0,j}(\rho^{1/p}), \rho^{1/\hatp}\big\rangle.$$ By using similar computations, we can also show that  $\big\langle \cL_{1,j}(\rho'^{1/p}), \rho'^{1/\hatp}\big\rangle = \big\langle \cL_{1,j}(\rho^{1/p}), \rho^{1/\hatp}\big\rangle$. On the other hand, $S(\rho') = S(D_\xi \rho D_\xi^\dagger) = S(\rho)$. Putting these together, we find that 
$$\Upsilon_m(\rho') = \Upsilon_m(\rho)+ \frac {\omega}m \sum_j |\xi_j|^2\geq \Upsilon_m(\rho). $$

\medskip
\noindent
(ii) For any passive transformation $U$, there is an $m\times m$ \emph{unitary} matrix $(u_{jk})_{j,k}$ such that $U^\dagger \bfa_j U = \sum_{k} u_{jk} \bfa_{k}$. Employing this relation yields
$$\sum_{j=1}^m\tr(\rho' \bfa_j^\dagger \bfa_j)=\sum_{j=1}^m\tr\big(\rho\, U^\dagger\bfa_j^\dagger U U^\dagger\bfa_j U\big)=\sum_{j,k,k'=1}^m u_{jk}\bar{u}_{jk'}\, \tr\big(\rho \bfa_k^\dagger \bfa_{k'}\big)=\sum_{k=1}^m\tr(\rho \bfa_k^\dagger \bfa_k),$$
where we used $\sum_{j}u_{jk}\bar{u}_{jk'}=\delta_{k,k'}$. Applying similar computations, one can also verify that $\big\langle \cL_{b,j}(\rho'^{1/p}), \rho'^{1/\hatp}\big\rangle = \big\langle \cL_{b,j}(\rho^{1/p}), \rho^{1/\hatp}\big\rangle$ for $b\in\{0,1\}$. Therefore, using these relations together with $S(\rho') = S(\rho)$, we obtain $\Upsilon_m(\rho')=\Upsilon_m(\rho)$.

\medskip
\medskip
\noindent
(iii) The proof of this part is more involved. First, we note that $S_r^\dagger \bfa_j S_r = \cosh(r_j)\bfa_j - \sinh(r_j)\bfa_j^\dagger$, which using $\cosh^2(r_j) - \sinh^2(r_j)=1$ implies
\begin{align}\label{eq:S-r-a}
S_r^\dagger \bfa_j^\dagger \bfa_j S_r = \bfa_j^\dagger \bfa_j + \sinh^2(r_j)(\bfa_j^\dagger \bfa_j + \bfa_j\bfa_j^\dagger) - \cosh(r_j)\sinh(r_j) \big((\bfa_j^\dagger)^2 + \bfa_j^2\big).
\end{align}
Using these relations and applying the assumption on $\rho$, we find that
\begin{align*}
 \big\langle \cL_{0,j}\big(\rho'^{1/p}\big), \rho'^{1/\hatp}\big\rangle &= \big\langle \cL_{0,j}\big(\rho^{1/p}\big), \rho^{1/\hatp}\big\rangle \\
 &\quad \, + \sinh^2(r_j)\Big( \tr(\rho \bfa_j^\dagger \bfa_j)  + \tr(\rho\bfa_j\bfa_j^\dagger) - \tr(\rho^{1/p} \bfa_j \rho^{1/\hatp} \bfa_j^\dagger) - \tr(\rho^{1/p}\bfa_j^\dagger \rho^{1/\hatp} \bfa_j)   \Big).
\end{align*}
Let $\rho=\sum_k \lambda_k\ketbra{\psi_k}{\psi_k}$ be the eigen-decomposition of $\rho$. Then, by Young's inequality we have
\begin{align*}
\tr(\rho^{1/p} \bfa_j \rho^{1/\hatp} \bfa_j^\dagger) + &\tr(\rho^{1/p}\bfa_j^\dagger \rho^{1/\hatp} \bfa_j) \\
& = \sum_{k, \ell} \Big( \lambda_k^{1/p}\lambda_\ell^{1/\hatp}+ \lambda_k^{1/\hatp}\lambda_\ell^{1/p}\Big)\big| \bra{\psi_k} \bfa_j \ket{\psi_\ell}\big|^2  \\
&\geq \sum_{k, \ell} \Big( \frac 1p\lambda_k + \frac 1{\hatp}\lambda_\ell+ \frac 1{\hatp}\lambda_k+ \frac 1p\lambda_\ell\Big)\big| \bra{\psi_k} \bfa_j \ket{\psi_\ell}\big|^2\\
&= \sum_{k, \ell} \big( \lambda_k + \lambda_\ell\big)\big| \bra{\psi_k} \bfa_j \ket{\psi_\ell}\big|^2\\
&= \sum_{k}  \lambda_k \big\| \bfa_j^\dagger \ket{\psi_k}\big\|^2 + \sum_{ \ell}  \lambda_\ell\big\|  \bfa_j \ket{\psi_\ell}\big\|^2\\
& = \tr\big(\rho \bfa_j\bfa_j^\dagger\big) + \tr\big(\rho \bfa_j^\dagger\bfa_j\big).
\end{align*}
Therefore, we can see that 
\begin{align*}
 \big\langle \cL_{0,j}\big(\rho'^{1/p}\big), \rho'^{1/\hatp}\big\rangle \geq  \big\langle \cL_{0,j}\big(\rho^{1/p}\big), \rho^{1/\hatp}\big\rangle,
 \end{align*}
and similarly $ \big\langle \cL_{1,j}\big(\rho'^{1/p}\big), \rho'^{1/\hatp}\big\rangle \geq  \big\langle \cL_{1,j}\big(\rho^{1/p}\big), \rho^{1/\hatp}\big\rangle$. Moreover, once again using~\eqref{eq:S-r-a} we have
$$\tr(\rho' H_m) =\tr(\rho H_m) + \sum_j \sinh^2(r_j) \big( \tr(\rho\bfa_j^\dagger \bfa_j) + \tr(\rho\bfa_j\bfa_j^\dagger)  \big) \geq \tr(\rho H_m).$$
We also have $S(\rho')=S(\rho)$. Putting these together we arrive at $\Upsilon_m(\rho')\geq \Upsilon_m(\rho)$.

\end{proof}

\begin{theorem}\label{thm:main-technical-result-multimode-Gaussian}
For any $\nu_0, \nu_1, \omega\geq 0$ let $\eta_\thermal$ be given by~\eqref{eq:def-eta-thermal}. Then, for any multimode state $\rho'=U_{\rm G}\, \rho U_{\rm G}^\dagger$ where $U_{\rm G}$ is an $m$-mode Gaussian unitary, we have
\begin{align}\label{eq:Upsilon-ineq-Gaussian}
\Upsilon_m(\rho')\geq \eta_\thermal,
\end{align}
assuming that one of the following conditions is satisfied:
\begin{itemize}
\item[{\rm (a)}] $\rho=\rho_{1}\otimes\cdots\otimes\rho_{m}$ is a product state and $U_G$ is passive.
\item[{\rm (b)}] $\rho=\rho_{\diag}$ is diagonal in the Fock basis. 
\end{itemize}

\end{theorem}

\begin{proof}
We first note that by the definition of $\Upsilon_m(\cdot)$ and using Theorem~\ref{thm:main-technical-result} we have
$$\Upsilon(\rho_1\otimes\cdots \otimes \rho_m) = \frac 1m \sum_j \Upsilon(\rho_j)\geq \eta_{\thermal}.$$
Then, (a) follows from part (ii) of Lemma~\ref{lem:Gaussian-unitaries}. To prove (b) we use the \emph{Bloch-Messiah decomposition} to express the Gaussian unitary $U_G$ as $U_{\rm G}=D_\xi U S_r V$, where $D_\xi$ is a mutlimode displacement operator, $U, V$ are passive Gaussian transformations, and $S_r$ is a product of single-mode squeezing transformations~\cite{SerafiniBook}. Therefore, the multimode state can be written as
$$\rho' = D_\xi U S_r V \rho_{\diag} V^\dagger S_r^\dagger U^\dagger D_\xi^\dagger,$$
where $\rho_{\diag}$ is diagonal in the Fock basis. We use this expression to prove~\eqref{eq:Upsilon-ineq-Gaussian} in the following steps:

\begin{itemize}
\item[-] By Theorem~\ref{thm:main-technical-result}, the meta log-Sobolev inequality holds for all single-mode diagonal states. Then, applying a standard tensorization argument in the classical case, i.e., using the chain rule and the concavity of the entropy function (see, e.g.,~\cite[Proposition 3.7]{Mossel+2013}), we find that $\Upsilon_m(\rho_\diag)\geq \eta_{\thermal}$.   Then, by part (ii) of Lemma~\ref{lem:Gaussian-unitaries}, we have $\Upsilon_m(V \rho_{\diag} V^\dagger)=\Upsilon_m(\rho_{\diag})\geq \eta_\thermal$.

\item[-] Next we use part (iii) of Lemma~\ref{lem:Gaussian-unitaries} to establish~\eqref{eq:Upsilon-ineq-Gaussian} for $S_r V \rho_{\diag} V^\dagger S_r^\dagger$. To this end, we need to verify the required assumptions $\tr\big(V \rho_{\diag} V^\dagger \bfa_j^2\big)=\tr\big( (V \rho_{\diag} V^\dagger)^{1/p} \bfa_j (V \rho_{\diag} V^\dagger)^{1/\hatp} \bfa_j\big)=0$.  We note that $\tr\big(V \rho_{\diag} V^\dagger \bfa_j^2\big) = \tr\big(\rho_{\diag} (V^\dagger \bfa_j V)^2\big)$ and $V^\dagger \bfa_j V= \sum_{k} v_{jk} \bfa_{k}$ is a linear combination of $\bfa_{k}$'s. Moreover, $\tr(\rho_{\diag} \bfa_k\bfa_{k'})=0$ for any $k, k'$ simply because $\rho_{\diag}$ is diagonal in the Fock basis. Therefore, we have  $\tr\big(V \rho_{\diag} V^\dagger \bfa_j^2\big)=0$. By the same argument, we also have $\tr\big( (V \rho_{\diag} V^\dagger)^{1/p} \bfa_j (V \rho_{\diag} V^\dagger)^{1/\hatp} \bfa_j\big)=0$. Hence, $\Upsilon_m(S_r V \rho_{\diag} V^\dagger S_r^\dagger)=\Upsilon_m(V \rho_{\diag} V^\dagger)=\Upsilon_m(\rho_{\diag})\geq \eta_\thermal$.

\item[-] Once again, using part (ii) of Lemma~\ref{lem:Gaussian-unitaries} we find that~\eqref{eq:Upsilon-ineq-Gaussian} holds for $U S_r V \rho_{\diag} V^\dagger S_r^\dagger U^\dagger$ since we have already verified it for $S_r V \rho_{\diag} V^\dagger S_r^\dagger$ and $U$ is a passive transformation.

\item[-] Finally,  we note that  $V^\dagger S_r^\dagger U^\dagger \bfa_j U S_r V$ is a linear combination of $\bfa_k$'s and $\bfa_k^\dagger$'s. Moreover, $\tr(\rho_{\diag} \bfa_k) =\tr(\rho_{\diag} \bfa_k^\dagger)=0$. Therefore, $\tr(U S_r V \rho_{\diag} V^\dagger S_r^\dagger U^\dagger \bfa_j)=0$ for all $j$. Thus, by part (i) of Lemma~\ref{lem:Gaussian-unitaries}, inequality~\eqref{eq:Upsilon-ineq-Gaussian} holds for $\rho = D_\xi U S_r V \rho_{\diag} V^\dagger S_r^\dagger U^\dagger D_\xi^\dagger$.
\end{itemize}
\end{proof}

Notice that Theorem~\ref{thm:main-technical-result-multimode-Gaussian} implies $\Upsilon_m(\rho)\geq \eta_{\thermal}$ for all multimode Gaussian states. The point is that by \emph{Williamson's theorem} any Gaussian state can be transformed into a tensor product of thermal states using a Gaussian unitary, and thermal states are diagonal in the Fock basis~\cite{SerafiniBook}.

%************************************************************************
\section{Log-Sobolev inequalities for the quantum Ornstein-Uhlenbeck semigroup}\label{sec:LSI-OU}

As discussed in Section~\ref{sec:phase-cov-channels}, the semigroup of attenuator channels is sometimes called the quantum (bosonic) Ornstein-Uhlenbeck semigroup.\footnote{The quantum Ornstein-Uhlenbeck semigroup restricted to a certain subspace of operators resembles the classical Ornstein-Uhlenbeck semigroup~\cite[Equation (7.5)]{Cipriani+2000}.}
In this section, we explicitly compute the optimal $p$-log-Sobolev constant for this semigroup for any $1\leq p\leq 2$, and derive a quantum variant of the celebrated log-Sobolev inequality of Gross~\cite{Gross75}.  Due to the equivalence of log-Sobolev inequalities and hypercontractivity inequalities for quantum Markov semigroups~\cite{OZ99, KT13, BDR20}, our results provide the optimal hypercontractivity inequalities for the quantum Ornstein-Uhlenbeck semigroup.

We need to develop some notations to %explain
present our results. Let $\Phi_t=e^{-t\cL}$ be the quantum attenuator channel with
\begin{align}\label{eq:L-OU}
\cL(\rho) = e^{-\beta/2}\Big( \frac 12 \big\{\bfa\bfa^\dagger, \rho\big\} - \bfa^\dagger \rho\bfa\Big) +  e^{\beta/2}\Big(  \frac 12 \big\{\bfa^\dagger\bfa, \rho\big\} - \bfa \rho\bfa^\dagger\Big),
\end{align}
where $\beta>0$ is some parameter. 
The adjoint of the channel with respect to the Hilbert-Schmdit inner product is denoted by $\Phi_t^*$ and describes the evolution in the Heisenberg picture. Then, $\Phi_t^*=e^{-t\cL^*}$ where\footnote{In the literature of log-Sobolev inequalities usually $\cL$ is the generator in the Heisenberg picture and $\cL^*$ is the generator in the Schr\"odinger picture. Here, we change the notation since we started with channels in the Schr\"odinger picture in previous sections.}
\begin{align*}
\cL^*(X) = e^{-\beta/2}\Big( \frac{1}{2}\big\{\bfa \bfa^\dagger, X\big\}-\bfa X\bfa^\dagger \Big)+e^{\beta/2}\Big( \frac{1}{2}\big\{\bfa^\dagger \bfa, X\big\} -\bfa^\dagger X\bfa\Big).
\end{align*}
The semigroup $\{\Phi_t:\, t\geq 0\}$ has a fixed point which we denote by $\sigma=\sigma_\beta$:
\begin{align*}
\sigma=\sigma_\beta =  (1-e^{-\beta})\sum_{n=0}^\infty e^{-\beta n}\ketbra nn.
\end{align*}
Then, it is natural to define a \emph{weighted} inner product with respect to this state:
\begin{align}\label{eq:weighted-inner-prod}
\langle X, Y\rangle_{\sigma} = \tr\big( \sigma^{1/2} X^\dagger \sigma^{1/2} Y\big)=\big\langle \Gamma_\sigma(X), Y\big\rangle,
\end{align}
where
$$\Gamma_\sigma(X) = \sigma^{1/2} X\sigma^{1/2}.$$
We emphasize that the weighted inner product $\langle \cdot, \cdot\rangle_\sigma$ should not be confused with the Hilbert-Schmidt inner product $\langle \cdot, \cdot\rangle$ that has no subscript.
This inner product induces the $2$-norm 
$$\|X\|_{2, \sigma}^2 =\tr\big( \big|\Gamma_\sigma^{1/2}(X)\big|^2 \big) =\tr\big( |\sigma^{1/4}X\sigma^{1/4}|^2 \big),$$ 
where $|Y| = \sqrt{Y^\dagger Y}$. It can be verified that $\cL=\Gamma_\sigma\circ \cL^*\circ \Gamma_\sigma^{-1}$ and that $\cL^*$ is self-adjoint with respect to this weighted inner product. Thus, we may consider the \emph{Dirichlet form} associated to $\cL^*$ given by
$$\mathcal E_p(\rho) = \frac{p\hatp}{4}\Big\langle  \Gamma_\sigma^{-\frac 1\hatp}\big(\rho^{\frac 1\hatp}\big) , \cL^*\circ\Gamma_\sigma^{-\frac 1p}\big(\rho^{\frac 1 p}\big)    \Big\rangle_{\sigma}.$$
Here, $p\geq 1$ and the case of $p=1$ is understood in the limit as
$$\mathcal E_1(\rho) = \frac{1}{4}\big\langle  \log \rho-\log \sigma , \cL^*\circ\Gamma_\sigma^{- 1}(\rho)    \big\rangle_{\sigma}.$$
Now, a $p$-log-Sobolev inequality with parameters $c\geq 0$  takes the form:
$$c D(\rho\|\sigma) \leq \mathcal E_p(\rho), \quad \forall \rho, $$
where $D(\rho\| \sigma)=\tr(\rho\log \rho) - \tr(\rho\log \sigma)$ is Umegaki's relative entropy. The optimal constant $c$ for which the above inequality holds is usually denoted by $\alpha_p$, i.e.,
$$\alpha_p = \inf_{\rho}  \frac{\mathcal E_p(\rho)}{D(\rho\| \sigma)},$$
where the infimum is taken over all states with finite mean photon number. We note that, as it will become clear in the proof of the following theorem, similar to $\Upsilon(\rho)$, $\cE_p(\rho)$ is well-defined for any state $\rho$ with finite mean photon number.

\begin{theorem}\label{thm:p-LSI}
For any $\beta>0$, the $p$-log-Sobolev constant of the quantum Ornstein-Uhlenbeck semigroup is given by
\begin{align}\label{eq:alpha-p}
\alpha_p = \frac{p\hatp}{4\beta} e^{\beta/2}\big(1-e^{-\beta/p}\big)\big(1-e^{-\beta/\hatp}\big),
\end{align}
if $p>1$ and $\alpha_1 = \lim_{p\to 1^+}\alpha_p = \frac 12 \sinh(\beta/2)$.
\end{theorem}

\begin{proof}
We note that $\sigma^s\bfa = e^{s\beta} \bfa \sigma^s$ for every $s$, and $\sigma$ commutes with $\bfa^\dagger\bfa$. Then, for a single-mode state $\rho$ we have
\begin{align}\label{eq:F-p-expansion}
\mathcal E_p&(\rho) =  \frac{p\hatp}{4}\Big\langle  \Gamma_\sigma^{-\frac 1\hatp}\big(\rho^{\frac 1\hatp}\big) , \cL^*\circ\Gamma_\sigma^{-\frac 1p}\big(\rho^{\frac 1 p}\big)   \! \Big\rangle_{\sigma} \nonumber \\
& = \frac{p\hatp}{4} \tr\Big(  \rho^{\frac1\hatp} \Gamma_\sigma^{\frac1p} \circ \cL^*\circ\Gamma_\sigma^{-\frac1p} (\rho^{\frac1p}) \Big) \nonumber \\
& = \frac{p\hatp}{4} \bigg(\!e^{-\beta/2}\Big[ \tr\big(\rho\,\bfa\,\bfa^\dagger\big) - \tr\big( \rho^{\frac1\hatp} \sigma^{\frac1{2p}} \bfa\, \sigma^{-\frac1{2p}}  \rho^{\frac1p} \sigma^{-\frac1{2p}}\bfa^\dagger  \sigma^{\frac1{2p}}  \big)   \Big] \nonumber  \\
&\qquad\quad\qquad+ e^{\beta/2}\Big[ \tr\big(\rho\,\bfa^\dagger\bfa\big) - \tr\big( \rho^{\frac1\hatp} \sigma^{\frac1{2p}} \bfa^\dagger \sigma^{-\frac1{2p}} \rho^{\frac1p} \sigma^{-\frac1{2p}}\bfa\,  \sigma^{\frac1{2p}}  \big)   \Big]
\!\bigg) \nonumber \\
& = \frac{p\hatp}{4} \bigg(\!e^{-\beta/2}\Big[ \tr\big(\rho\,\bfa\,\bfa^\dagger\big) - e^{\frac\beta p} \tr\big( \rho^{\frac1\hatp}  \bfa\, \rho^{\frac1p} \bfa^\dagger    \big)   \Big] + e^{\beta/2}\Big[ \tr\big(\rho\,\bfa^\dagger\bfa\big) - e^{-\frac\beta p}\tr\big( \rho^{\frac1\hatp}  \bfa^\dagger  \rho^{\frac1p} \bfa    \big)   \Big]
\!\bigg) \nonumber \\
& = \hatp \Big( \nu'_0\big\langle \cL_0(\rho^{1/p}), \rho^{1/\hatp} \big\rangle + \nu'_1\big\langle \cL_1\big(\rho^{1/p}\big), \rho^{1/\hatp} \big\rangle   \Big) + \omega'\tr\big(\rho\,\bfa^\dagger\bfa\big) + \frac{p\hatp}{4}\Big(e^{-\beta/2} - e^{-\big(\frac 12 -\frac 1p\big)\beta}\Big),
\end{align}
where $\cL_0, \cL_1$ are given in~\eqref{eq:L0-L1} and
\begin{align*}
\nu'_0 = \frac p4 e^{\big(\frac 12 -\frac 1p\big)\beta}, \qquad\qquad \nu'_1=\frac p4e^{-\big(\frac 12 -\frac 1p\big)\beta}, 
\end{align*}
\begin{align*}
\omega'=\frac{p\hatp}{4}\Big(e^{\beta/2} + e^{-\beta/2} -  e^{\big(\frac 12 -\frac 1p\big)\beta}-e^{-\big(\frac 12 -\frac 1p\big)\beta} \Big)=\frac{p\hatp}{4}e^{\beta/2}\big(1-e^{-\beta/p}\big)\big(1-e^{-\beta/\hatp}\big).
\end{align*}
On the other hand, since $\sigma = (1-e^{-\beta}) e^{-\beta\bfa^\dagger \bfa}$, we have 
$$D(\rho\| \sigma) =\tr(\rho\log \rho) - \tr(\rho\log \sigma) = -S(\rho) +\beta\tr(\rho\,\bfa^\dagger \bfa) - \log(1-e^{-\beta}).$$
Putting these together, we find that for $\alpha_p$ given by~\eqref{eq:alpha-p} we have
$$\frac{1}{\alpha_p}\mathcal E_p(\rho)- D(\rho\| \sigma) -\log(1-e^{-\beta})-\frac{p\hatp}{4}\Big(e^{-\beta/2} - e^{-\big(\frac 12 -\frac 1p\big)\beta}\Big)=\Upsilon(\rho),$$
where $\Upsilon(\rho)$ is given in~\eqref{eq:def-Upsilon}
with $\nu_0 = \frac{\nu'_0}{\alpha_p}$, $\nu_1=\frac{\nu'_1}{\alpha_p}$ and $\omega=\frac{\omega'}{\alpha_p}-\beta=0$.   
Therefore, to prove the theorem we can use Theorem~\ref{thm:main-technical-result} to conclude that
it suffices to show that for any thermal state $\tau$ we have 
\begin{align}\label{eq:p-LSI-proof}
\alpha_p D(\tau\| \sigma) \leq  \mathcal E_p(\tau),
\end{align}
where $\alpha_p$ is given by~\eqref{eq:alpha-p}, and it is the best possible such constant.

\medskip

By the above computations we have 
\begin{align*}
 \mathcal E_p(\tau) & = \frac{p\hatp}{4} \bigg(\!e^{-\beta/2}\Big[ \tr(\tau\bfa\,\bfa^\dagger) - e^{\frac\beta p} \tr\big( \tau^{\frac1\hatp}  \bfa\, \rho^{\frac1p} \bfa^\dagger    \big)   \Big] + e^{\beta/2}\Big[ \tr(\tau\bfa^\dagger\bfa) - e^{-\frac\beta p}\tr\big( \tau^{\frac1\hatp}  \bfa^\dagger  \tau^{\frac1p} \bfa    \big)   \Big]\!\bigg)\\
 & = \frac{p\hatp}{4} \bigg(\!  \big(e^{\beta /2}+e^{-\beta/2}\big) \tr(\tau\bfa^\dagger \bfa)  +e^{-\beta/2}   - e^{\beta\big(\frac1 p-\frac 12\big)} \tr\big( \tau^{\frac1\hatp}  \bfa \tau^{\frac1p} \bfa^\dagger    \big)   - e^{\beta\big(\frac1 {\hatp} - \frac 12\big)}\tr\big( \tau^{\frac1\hatp}  \bfa^\dagger  \tau^{\frac1p} \bfa    \big) \!  \bigg).
\end{align*}
Here, for convenience we use a new parametrization for thermal states:
$$\tau = (1-y^2)\sum_{n} y^{2n}\ketbra{n}{n}.$$
%Then
Using this, we have
\begin{align*}
&\tr(\tau\bfa^\dagger \bfa) = \frac{y^2}{1-y^2},  \qquad\qquad \tr(\tau\log \tau) = \frac{y^2}{1-y^2}\log y^2+ \log(1-y^2),\\
&\tr\big( \tau^{\frac1 p}  \bfa^\dagger  \tau^{\frac1\hatp} \bfa    \big)= \frac{y^{\frac 2p}}{1-y^2} , \qquad \tr\big( \tau^{\frac1\hatp}  \bfa^\dagger  \tau^{\frac1p} \bfa    \big)  = \frac{y^{\frac 2\hatp}}{1-y^2}.
\end{align*}
Therefore, we get
\begin{align*} 
\mathcal E_p(\tau)  & = \frac{p\hatp}{4(1-y^2)} \bigg( \! \Big(e^{\beta /2}+e^{-\beta/2}\Big) y^2 +e^{-\beta/2}(1-y^2)   - e^{\beta\big(\frac1 p-\frac 12\big)} y^{\frac 2p}   - e^{\beta\big(\frac1 {\hatp} - \frac 12\big)}y^{\frac 2\hatp}\bigg)\\
&= \frac{p\hatp}{4(1-y^2)}  \bigg(\!  e^{\beta /2} y^2+e^{-\beta/2}   - e^{\beta\big(\frac1 p-\frac 12\big)} y^{\frac 2p}   - e^{\beta\big(\frac1 {\hatp} - \frac 12\big)}y^{\frac 2\hatp}\bigg)\\
&= \frac{p\hatp}{4(1-y^2)}  e^{\beta/2} \Big( y^{\frac 2p}-e^{-\frac\beta p} \Big)\Big(  y^{\frac 2\hatp}-e^{-\frac\beta {\hatp} }\Big).
\end{align*}
We also have
\begin{align*}
D(\tau\| \sigma)& = \tr(\tau\log \tau) + \beta\tr(\tau\bfa^\dagger \bfa) - \log(1-e^{-\beta})\\
& = \frac{1}{1-y^2}\Big( y^2\log y^2 + (1-y^2)\log(1-y^2) +\beta y^2 -(1-y^2)\log(1-e^{-\beta})  \Big)\\
& = \frac{1}{1-y^2} d(y^2\| e^{-\beta}),
\end{align*}
where $d(y^2\|x^2) = y^2\log y^2 + (1-y^2)\log(1-y^2) -y^2\log x^2-(1-y^2)\log (1-x^2) $ is the binary relative entropy function. 
Hence,~\eqref{eq:p-LSI-proof} is equivalent to $\phi(x, y)\geq 0$ for all $0<x, y<1$ where $x^2=e^{-\beta}$ and 
\begin{align*}
\phi(x, y) & =  (1-y^2)\frac{\mathcal E_p(\tau)}{\alpha_p} -  (1-y^2)D(\tau\| \sigma)  \\
&= -\frac{\big( x^{\frac 2p}-y^{\frac2 p} \big)\big(  x^{\frac 2\hatp}-y^{\frac2 {\hatp} }\big)\log x^2}{\big(1-x^{\frac 2p}\big)\big(1-x^{\frac 2\hatp}\big)} - d(y^2\| x^2).
\end{align*}
Now the idea is to fix $y$ and think of $\phi(x, y)$ as a function of $x$. It is shown in Appendix~\ref{app:comp-derivative} that 
\begin{align}\label{eq:phi-derivative}
 \frac{\dd}{\dd x}\phi(x, y)  = &\,   \frac{2x\Big(\!  \big(1- x^{\frac 2\hatp}\big)\big(1-y^{\frac 2\hatp}\big) \big(x^{\frac 2p} - y^{\frac 2p}\big)   + \big(1- x^{\frac 2p}\big)\big(1-y^{\frac 2p}\big) \big(x^{\frac 2\hatp} - y^{\frac 2\hatp}\big) \!  \Big)}{(1-x^2)\big(1-x^{\frac 2p}\big)\big(1-x^{\frac 2\hatp}\big)} \nonumber\\
&~ - \frac{2x^{\frac 2p}\big(1-y^{\frac 2p}\big)\big(x^{\frac 2\hatp} - y^{\frac 2\hatp}\big)}{x\big(1-x^{\frac 2p}\big)\big(1-x^{\frac 2\hatp}\big)}  \bigg(  \frac{\log x^2}{p(1-x^{\frac 2p})}+1 \bigg)  \nonumber \\
&~ - \frac{2x^{\frac 2\hatp}\big(1-y^{\frac 2\hatp}\big)\big(x^{\frac 2p} - y^{\frac 2p}\big)}{ x\big(1-x^{\frac 2p}\big)\big(1-x^{\frac 2\hatp}\big)} \bigg(  \frac{\log x^2}{\hatp(1-x^{\frac 2\hatp})}+1 \bigg) .
\end{align}
We argue that this derivative vanishes on the interval $(0,1)$ only if $x=y$. To this end, we use $\log t< t-1$ for $1\neq t\in \big\{x^{\frac 2p}, x^{\frac 2\hatp}\big\}$ to conclude that for $0<x<1$,
$$-\bigg(\frac{\log x^2}{p(1-x^{\frac 2p})}+1\bigg) , -\bigg(\frac{\log x^2}{\hatp(1-x^{\frac 2\hatp})}+1\bigg)>0.$$
Therefore, $ \frac{\dd}{\dd x}\phi(x, y)$ is non-zero for $0<x\neq y<1$ and its sign depends on whether $x>y$ or $x<y$.
In fact, if $0<x_1<y<x_2<1$, then
$$ \frac{\dd}{\dd x}\phi(x, y)  \Big|_{x=x_1} <  \frac{\dd}{\dd x}\phi(x, y)  \Big|_{x=y}=0 < \frac{\dd}{\dd x}\phi(x, y)  \Big|_{x=x_2}.$$
Thus, the minimum of $\phi(x, y)$, as a function of $x$, is achieved at $x=y$ and we have $\phi(y, y)=0$. This means that $\phi(x, y)\geq 0$ for all $0<x, y<1$, and~\eqref{eq:p-LSI-proof} holds for any thermal state $\tau$. Also, the limiting case of $y\to 1^-$ confirms that the constant $\alpha_p$ in~\eqref{eq:p-LSI-proof} is optimal.

\end{proof}

We remark that some authors define the quantum Ornstein-Uhlenbeck semigroup by considering arbitrary parameters $\nu_0, \nu_1$ satisfying $\nu_1>\nu_0>0$, while here we study only the case of $\nu_0=\nu_1^{-1}=e^{-\beta/2}$. Nevertheless, by a rescaling argument all the log-Sobolev constants of these semigroups can also be computed using Theorem~\ref{thm:p-LSI}.

Extending the definition of the weighted $2$-norm, for any $p\geq 1$ we may define
\begin{align}\label{eq:w-p-norm}
\|X\|_{p,  \sigma} = \tr\Big( \Big|  \Gamma_{ \sigma}^{\frac{1}{p}}(X) \Big|^{p} \Big)^{\frac 1p}.
\end{align}
We let $L_p(\sigma)$ to be the closure of  the space of bounded operators under this norm.

The following corollary is a consequence of the above theorem, and~\cite[Theorem 3.8]{OZ99} (see also~\cite[Theorem 11]{BDR20}).

\begin{corollary}
Let $\Phi_{t} = e^{-t\cL}$ where $\cL$ is given by~\eqref{eq:L-OU}. Also, let $\Phi^*_{t} = e^{-t\cL^*}$.  
Then, for any $1< q\leq p<  +\infty$ and operator $X\in L_q(\sigma)$ we have
$$\big\|\Phi^*_{t}(X)\big\|_{p,  \sigma}\leq  \big\|X\big\|_{q,\sigma}, \qquad \quad \forall t\geq  \frac{1}{4\alpha_{2}} \log \frac{p-1}{q-1},$$
where  $\alpha_2= \frac{4}{\beta}\sinh^2(\beta/4)$.
\end{corollary}

\begin{proof}
It is well-known that it suffices restrict to operators that are positive semidefinite (see, e.g.,~\cite{BDR20}). For $p\geq q$ define $t(p):= \frac{1}{4\alpha_2} \log \frac{p-1}{q-1}$ and note that $t(q)=0$. For a positive semidefinite operator $X$ let $X_p:= \Phi_{t(p)}^*(X)$ and define $f(p):=\|X_p\|_{p, \sigma} - \|X\|_{q, \sigma}$. Our goal is to show that  $f(p)\leq 0$ for any $p\geq q$. To this end, since $f(q)=0$, it suffices to verify that $f'(p)\leq  0$. Computing the derivative (see the details in the proof of~\cite[Theorem 11]{BDR20}) we find that
$$f'(p)=\frac{1}{p^2}\|X_p\|_{p, \sigma}^p \Big( D(\rho_p\| \sigma) - \frac{1}{\alpha_2}\cE_p(\rho_p)\Big),$$
where 
$$\rho_p := \frac{1}{\|X_p\|_{p, \sigma}^p} \Gamma_\sigma^{1/p}(X_p)^p.$$
It is shown in Appendix~\ref{app:alpha2-min} that  $ \alpha_p\geq \alpha_2$ for all $p$. Then, the negativity of $f'(p)$ follows from the log-Sobolev inequality of Theorem~\ref{thm:p-LSI}.

For the above argument we need to make sure that $\Phi_{t}^*(X)$ is a well-defined operator. To this end, we notice that  the quantum Ornstein-Uhlenbeck semigroup satisfies the \emph{Feller property} with respect to the algebra of compact operators~\cite[Theorem 5.1]{Cipriani+2000}. In particular, $\Phi_t^*(X)$ is a well-defined compact operator if $X$ is compact. Thus, as also done in the proof of~\cite[Theorem 3.8]{OZ99}, in the above argument we may restrict to compact operators $X$, and after proving the inequality for such an operator, generalize to arbitrary  $X\in L_q(\sigma)$ by a continuity argument. We also note that since $\sigma$ is Gaussian, when $X$ is compact and then bounded, the corresponding density operator $\rho_p$ in the above argument has a finite mean photon number. Thus, the log-Sobolev inequality of Theorem~\ref{thm:p-LSI} can indeed be employed.

\end{proof}

\subsection{Log-Sobolev inequality for multimode states}\label{subsec:LSI-multimode}

It is well-known that log-Sobolev constants for classical Markov semigroups satisfy the tensorization property, meaning that for generators $\cK_1, \cK_2$ of two classical Markov semigroups we have $\alpha_p(\cK_1\otimes  I +  I\otimes \cK_2)= \min\{\alpha_p(\cK_1), \alpha_p(\cK_2)\}$, where $\cK_1\otimes \mathcal I + \mathcal I\otimes \cK_2$ is the generator of the tensor product semigroup $\{e^{-t\cK_1}\otimes e^{-t\cK_2}:\, t\geq 0\}$. This tensorization property is known only for some special quantum Markov semigroups; see~\cite{BDR20} and references therein for more details. Thus, a natural question is whether the tensorization property holds for the quantum Ornstein-Uhlenbeck semigroup. To explore this problem we first develop some notations.

For parameters $\beta_1, \dots, \beta_m>0$ let 
$$\widehat \cL = \cL_1+\cdots + \cL_m,$$
be the generator of an $m$-mode quantum Markov semigroup where
\begin{align}\label{eq:L-OU}
\cL_j(\rho) = e^{-\beta_j/2}\Big( \frac 12 \{\bfa_j\bfa_j^\dagger, \rho\} - \bfa^\dagger_j \rho\bfa_j\Big) +  e^{\beta_j/2}\Big(  \frac 12 \{\bfa_j^\dagger\bfa_j, X\} - \bfa_j \rho\bfa_j^\dagger\Big).
\end{align}
Then, the quantum channels
$$\widehat \Phi_t = e^{-t\widehat \cL} = e^{-t\cL_1}\otimes \cdots \otimes e^{-t\cL_m},$$
form an $m$-mode semigroup. The fixed point of this semigroup is 
$$\widehat \sigma = \sigma_{\beta_1}\otimes \cdots\otimes \sigma_{\beta_m},$$
and its corresponding Dirichlet form is equal to  
$$ {\mathcal E}_p(\rho) = \frac{p\hatp}{4} \Big\langle  \Gamma_{\widehat \sigma}^{-\frac 1\hatp}\big(\rho^{\frac 1\hatp}\big) , \widehat \cL^*\circ\Gamma_{\widehat \sigma}^{-\frac 1p}\big(\rho^{\frac 1 p}\big)    \Big\rangle_{\widehat \sigma}=\frac{p\hatp}{4} \sum_{j=1}^m \Big\langle  \Gamma_{\widehat \sigma}^{-\frac 1\hatp}\big(\rho^{\frac 1\hatp}\big) , \cL_j^*\circ\Gamma_{\widehat \sigma}^{-\frac 1p}\big(\rho^{\frac 1 p}\big)    \Big\rangle_{\widehat \sigma}.$$
Now, as in the single mode case we are interested in the $p$-log-Sobolev constant
$$\widehat \alpha_p = \inf_{\rho}  \frac{{\mathcal E}_p(\rho)}{D(\rho\| \widehat\sigma)}.$$

Assuming that the optimal states in the above optimization problem are Gaussian, the following corollary is an evidence for the tensorization property. 

\begin{corollary}\label{cor:LSI-Gaussian}
Suppose that $\beta_1=\cdots=\beta_m= \beta$. Then, for any $m$-mode Gaussian state $\rho$ we have
$$\alpha_{p} D(\rho\| \widehat \sigma)\leq \widehat{\mathcal E}_p(\rho),$$
where $\alpha_{p} = \frac{p\hatp}{4\beta}  e^{\beta/2}\big(1-e^{-\beta/p}\big)\big(1-e^{-\beta/\hatp}\big)$.
\end{corollary}

\begin{proof}
This is a direct consequence of Theorem~\ref{thm:main-technical-result-multimode-Gaussian} as well as Theorem~\ref{thm:p-LSI} and its proof. 

\end{proof}

Our next goal is to prove a lower bound on $\widehat \alpha_2$ which probably is not optimal, yet may be useful for some applications. To this end, it is useful to consider a correspondence between density operators $\rho$ and positive semidefinite operators $X$ as follows:
$$X=\Gamma_{\widehat \sigma}^{-1/2}(\rho^{1/2}).$$
Then, we note that $\|X\|_{2, \widehat \sigma}^2 = \tr(\rho)=1$. Indeed, $\Gamma_{\widehat \sigma}^{-1/2}$ is an isometric embedding of the space of Hilbert-Schmidt operators into $L_2(\widehat \sigma)$, the Hilbert space of operators equipped with norm $\|\cdot\|_{2, \widehat \sigma}$. With this correspondence, we have
$$\cE_2(\rho) = \big\langle X, \widehat\cL^*(X)\big\rangle_{\widehat \sigma}.$$
From this equation we realize that the spectrum of $\widehat\cL^*$ is a relevant object in the study of the Dirichlet form. 
Here, we consider $\widehat \cL^*$ as an operator acting on the Hilbert space $L_2(\widehat \sigma)$. 
We note that, since $\widehat \sigma$ is a Gaussian state,  any operator $X$ that is a polynomial of operators $\bfa_j, \bfa_j^\dagger$ belongs to $L_2(\widehat \sigma)$. Letting $\mathcal P$ be the space of these polynomial operators, we find that indeed $\mathcal P \subset L_2(\widehat \sigma)$ is a dense subset. Moreover, by definition, $\widehat \cL^*$ leaves this subspace invariant: $\widehat \cL^*(\mathcal P)\subseteq \mathcal P$. Thus, the domain of $\widehat \cL^*$ contains the dense subset $\mathcal P$~\cite{Cipriani+2000}. In the following proposition, we summarize Theorem~7.2 of~\cite{Cipriani+2000} regarding the spectrum of $\widehat \cL^*$ acting on $L_{2}(\widehat \sigma)$,  and for the sake of completeness present its main proof idea in Appendix~\ref{app:spectrum}.\footnote{In fact in~\cite{Cipriani+2000}, the spectrum of $\Gamma_{\widehat \sigma}^{1/2}\circ \widehat\cL^*\circ \Gamma_{\widehat \sigma}^{-1/2}$ as an operator acting on the space of Hilbert-Schmidt operators is computed.}

\begin{prop}\label{prop:spectrum1} \emph{\cite{Cipriani+2000}}
Consider $\widehat{\cL}^*$ as an operator acting on the Hilbert space $L_2(\widehat \sigma)$  equipped with norm $\|\cdot\|_{2, \widehat\sigma}$.
For $z_1, \dots, z_m\in \mathbb C$ with $|z_j|=1$ define the quadrature operator 
$\bfq_{j,z_j} = \frac {1}{\sqrt 2} (z_j\bfa_j^\dagger  + \bar z_j\bfa_j)$.
Also, for any $j$, let $\{h_{j,k_j}(t):\, k_j\geq 0\}$ be the set of Hermite polynomials specified by
\begin{align}\label{eq:hermit}
e^{st-\frac{\coth(\beta_j/2)}{4}s^2} = \sum_{k_j=0}^\infty \frac{s^{k_j}}{k_j!}h_{j,k_j}(t).
\end{align}
Then, for any tuple $k=(k_1, \dots, k_m)$ of non-negative integers, the operator $V_k:=\bigotimes_{j=1}^m h_{j,k_j}(\bfq_{j,z_j})$ is an eigenvector of $\widehat\cL^*$ with eigenvalue $\sum_{j=1}^m\sinh(\beta_j/2)k_j$. Moreover, these operators (and their appropriate linear combinations) are all the eigenvectors of $\widehat \cL^*$. 
\end{prop}

\medskip

For any $\ell=(\ell_1, \dots, \ell_m) \in \mathbb Z^m$, let $\mathcal F_\ell\subset L_2(\widehat \sigma)$ be the space of operators $X$ such that $\|X\|_{2, \widehat \sigma}<+\infty$ and for any $n=(n_1, \dots, n_m)$
their matrix entries in the Fock basis satisfy
$$\bra{n_1, \dots, n_m} X\ket{n_1+\ell'_1, \dots, n_m+\ell'_m}=0, \qquad \forall \ell'\neq \ell.$$
Then, any operator $X$ with $\|X\|_{2, \widehat \sigma}<+\infty$ can be decomposed as
\begin{align}\label{eq:diag-decomposition}
X= \sum_{\ell\in \mathbb Z^m} X_\ell, \qquad \quad X_\ell \in \mathcal F_\ell.
\end{align}
Following~\cite{CarboneSasso2008} we call~\eqref{eq:diag-decomposition} the \emph{diagonal decomposition} of $X$. A crucial property of the diagonal decomposition is that $\cF_\ell$'s are orthogonal subspaces, and $\langle X_\ell, X_{\ell'}\rangle_{\widehat \sigma}=0$ if $\ell\neq \ell'$. This is a consequence of the fact that $\widehat \sigma$ is diagonal in the Fock basis.

\begin{corollary} \textup{\cite{Cipriani+2000} } \label{cor:spectrum}
For any $\ell\in \mathbb Z^m$ the subspace $\mathcal F_\ell$ is invariant under $\widehat \cL^*$, and the eigenvalues of $\widehat \cL^*$ restricted to $\mathcal F_\ell$ are contained in $\big\{\sum_j \sinh(\beta_j/2)k_j:~ k_j\geq|\ell_j|    \big\}$. In particular, for any $X_\ell \in \mathcal F_\ell$ we have 
$$\bigg(\sum_{j=1}^m  \sinh(\beta_j/2)|\ell_j|  \bigg) \|X_\ell\|_{2, \widehat \sigma}\leq \big\langle X_\ell, \widehat \cL^*(X_\ell)\big\rangle_{2, \widehat \sigma}.$$
\end{corollary}

\begin{proof}
The fact that $\cL^*(\cF_\ell)\subseteq \cF_\ell$ is easily verified using the definition of $\widehat \cL^*$. Therefore, considering the diagonal decomposition $V_k= \sum_\ell V_{k, \ell}$ of the eigenvector $V_k$ defined in Proposition~\ref{prop:spectrum1}, we find that if $V_{k, \ell}\neq 0$, then it is also an eigenvector of $\widehat \cL^*$ with the same eigenvalue as that of $V_k$. On the other hand, by definition, $V_k$ is a polynomial of degree $k_j$ in terms of $\bfa_j$ and $\bfa_j^\dagger$. This means that we have $V_{k, \ell}=0$ if $|\ell_j|>k_j$ for some  $j$. We also note that since by Proposition~\ref{prop:spectrum1}, the closure of the span of $V_k$'s is the whole space $L_2(\widehat \sigma)$, the closure of the span of operators $V_{k, \ell}$ equals $\mathcal F_\ell$.  Indeed, $\mathcal F_\ell= \overline{\text{span}\{V_{k, \ell}:\, k_j\geq |\ell_j|, \forall j\}}$. Putting these together the desired result is implied.

\end{proof}

\medskip
To state our next lemma it is convenient to extend the definition of the entropy function for operators that are not necessarily normalized. For any positive operator $X=\Gamma_{\widehat \sigma}^{-1/2}(\rho^{1/2})$ where $\rho$ is a density operator, we define 
$$\Ent_{2, \widehat \sigma}(X) = D(\rho\| \widehat \sigma).$$
Extending this definition to non-normalized operators satisfying $\|X\|_{2, \widehat \sigma}<+\infty$, we let 
$$\Ent_{2, \widehat\sigma}(X) = \|X\|_{2, \widehat \sigma}^2 \,\Ent_{2, \widehat\sigma}\Big(\frac{X}{\|X\|_{2, \widehat\sigma}}\Big).$$
The significance of this entropy function is in its connection to the derivative of $p\mapsto \|X\|_{p, \widehat \sigma}$ at $p=2$ for which we refer to~\cite{OZ99, KT13, BDR20}.

\begin{lemma}\label{lem:entropic-ineq}
Let $X$ be a positive operator satisfying  $\tr\big(\rho H_m \big)<+\infty$ where   $ \rho= \Gamma_{\widehat \sigma}^{1/2}(X)^2 $ and $H_m=\sum_{j=1}^m \bfa_j^\dagger \bfa_j$ is the $m$-mode number operator.
Then, for any vector $w=(w_\ell)_\ell$ of positive real numbers we have
\begin{align}\label{eq:entropic-ineq}
\Ent_{2, \widehat \sigma}(X) \leq \sum_{\ell} \big( \log \|w\|_2^2 - \log w_{\ell}^2\big) \|X_\ell\|_{2, \widehat \sigma}^2 + \sum_\ell  \Ent_{2, \widehat \sigma}\big(I_{2,2}(X_\ell)\big) ,
\end{align}
where $X_\ell$'s are given by the diagonal decomposition $X= \sum_\ell X_\ell$ with $X_\ell\in \mathcal F_\ell$. Also, $I_{2,2}(X_\ell) = \Gamma_{\widehat \sigma}^{-1/2}\big(\big|\Gamma_{\widehat \sigma}^{1/2}(X_\ell) \big|\big)$ and  
 $\|w\|_2^2 = \sum_\ell w_\ell^2$.
\end{lemma}

\begin{proof}
 Fix a finite subset $S\subset \mathbb Z^m$ satisfying $-\ell \in S$ for all $\ell \in S$. For some technical reason, in the following we restrict to the subspace of operators $X$ satisfying $X_\ell=0$  if $\ell\notin S$. Later, we will relax this assumption.

For any $z\in \mathbb C$ define the map $T_z$ by
$$T_z(X) = \sum_{\ell} w_\ell^zX_\ell,$$
where as before $X=\sum_\ell X_\ell$ is the diagonal decomposition of $X$. Note that by the above assumption this sum is indeed a finite sum and over $\ell\in S$.   
By the orthogonality of the subspaces $\mathcal F_\ell$ for any $t\in \mathbb R$ we have
\begin{align}\label{eq:norm-it}
\big\|T_{it}(X)\big\|_{2, \widehat \sigma}^2 = \Big\|\sum_\ell  w_{\ell}^{it} X_\ell\Big\|_{2, \widehat \sigma}^2=
\sum_\ell |w_{\ell}^{it}|^2\cdot \|X_\ell\|_{2, \widehat \sigma}^2=
\sum_\ell \|X_\ell\|_{2, \sigma}^2.
\end{align}
Next, recall the definition of norm $\|\cdot\|_{p, \widehat \sigma}$ given by~\eqref{eq:w-p-norm}. Using the triangle and Cauchy-Schwarz inequalities, for any $t\in \mathbb R$ we have
\begin{align}\label{eq:norm-1+it}
\big\|T_{1+it}(X)\big\|_{\infty, \widehat \sigma}^2 & =\Big\|\sum_\ell w_{\ell}^{1+it} X_\ell \Big\|_{\infty, \widehat \sigma}^2 \nonumber\\
&\leq \Big( \sum_\ell |w_{\ell}^{1+it}|\cdot \|X_k\|_{\infty, \widehat  \sigma}\Big)^2\nonumber\\
&\leq \Big(\sum_\ell w_{\ell}^2\Big)\Big(  \sum_{\ell} \|X_\ell\|_{\infty, \widehat \sigma}^2 \Big).
\end{align}
For any $1\leq p\leq +\infty$ define
$$\|X\|_{2, p, \widehat \sigma} := \Big( \sum_\ell \|X_\ell\|^2_{p,\widehat \sigma} \Big)^{1/2}.$$
It is well-known and can be easily verified that $\|\cdot\|_{2, p, \widehat \sigma}$ satisfies the triangle inequality and is really a norm. Moreover, these norms form an \emph{interpolation family}~\cite{Pisier98}. 
Now, with this notation,~\eqref{eq:norm-it} and~\eqref{eq:norm-1+it} imply that
$$\|T_{it}\|_{(2, 2, \widehat \sigma)\to (2, \widehat \sigma)}=1, \qquad \|T_{1+it}\|_{(2, \infty, \widehat \sigma)\to (\infty, \widehat \sigma)}\leq \|w\|_2, \qquad \quad \forall t\in \mathbb R.$$
 Therefore, by the interpolation inequality~\cite{Lunardi2018}, we have
\begin{align*}
\|T_{1-2/p}\|_{(2, p, \sigma)\to (p, \sigma)}\leq \|w\|_2^{1-2/p}, \qquad \quad \forall ~2\leq p\leq +\infty.
\end{align*}
This means that 
\begin{align*}
\Big\|\sum_\ell w_{\ell}^{1-2/p} X_\ell \Big\|_p \leq \|w\|_2^{1-2/p}\cdot \Big(\sum_\ell \|X_\ell \|_p^2\Big)^{1/2}.
\end{align*}
We note that this inequality turns into an equality for $p=2$. Therefore, we must have  
$$\frac{\dd}{\dd p}\bigg(\Big\|\sum_\ell w_{\ell}^{1-2/p} X_\ell \Big\|_p \bigg)\bigg|_{p=2} \leq \frac{\dd}{\dd p}\bigg(\|w\|_2^{1-2/p}\cdot \Big(\sum_\ell \|X_\ell \|_p^2\Big)^{1/2} \bigg)\bigg|_{p=2}.
$$
Computing the derivative of both sides by~\cite[Proposition 3]{BDR20} and using the fact that $X_{\ell}^\dagger = X_{-\ell}$, the desired inequality is obtained.

Now we relax the assumption that $X_\ell$ is non-zero only for finitely many $\ell$. Assume that $X=  \|X\|_{2, \widehat \sigma}  \Gamma_{\widehat \sigma}^{-1/2}\big(\rho^{1/2}\big)$ where $\rho$ is a quantum state with finite mean photon number, in which case 
$$\Ent_{2, \widehat \sigma}(X) = \|X\|^2_{2, \widehat \sigma} D( \rho\| \widehat \sigma) =  -\|X\|^2_{2, \widehat \sigma}\big(S(\rho) + \tr(\rho \log \widehat \sigma)\big).$$ 
Let $\Pi_k$ be the projection on the span of basis vectors $\ket{n_1, \dots, n_m}$ satisfying $n_1+\cdots + n_m\leq k$. Let $X^{(k)}  = \Pi_k X \Pi_k$ and consider the quantum state $\rho^{(k)}=\|X^{(k)}\|^{-2}_{2, \widehat \sigma} \Gamma_{\widehat \sigma}^{-1/2}\big(X^{(k)}\big)^2$.
Observe that  $X_\ell^{(k)}$ is non-zero only for finitely many values of $\ell$. Then, by the above argument~\eqref{eq:entropic-ineq} holds for $X^{(k)}$. Taking the limit of $k\to +\infty$, clearly $\big\|X_\ell^{(k)}\big\|_{2, \widehat\sigma}$ tends to $\|X_\ell\|_{2, \widehat\sigma}$, which also implies that the $2$-norm of  $X^{(k)}$ tends to that of $X$. For the entropy terms, it can be verified that the mean photon number of $\rho^{(k)}$ tends to that of $\rho$. This implies that $\tr\big(\rho^{(k)} \log \widehat \sigma\big)$ tends to $\tr(\rho \log \widehat \sigma)$. 
Next, using the assumption $\tr(\rho H_m)<+\infty$, we can apply the continuity bound for the entropy function~\cite[Lemma 18]{winter2016tight} to conclude that $S\big(\rho^{(k)}\big)$ tends to $S(\rho)$. Therefore, 
$\Ent_{2, \widehat \sigma}\big(X^{(k)}\big)\to \Ent_{2, \widehat \sigma}(X)$ and similarly $\Ent_{2, \widehat \sigma}\big(I_{2,2}\big(X^{(k)}_\ell\big) \to\Ent_{2, \widehat \sigma}\big(I_{2,2}(X_\ell)$ as $k\to +\infty$. Putting these together the desired inequality for $X$ is implied.

\end{proof}

We remark that Lemma~\ref{lem:entropic-ineq} holds beyond the diagonal decomposition considered above and works for any decomposition of the space of operators into orthogonal subspaces. Thus, this lemma is of independent interest  and may find other applications.

\begin{theorem}\label{thm:multimode-LSI}
For any $\beta_1, \dots, \beta_m>0$ we have 
\begin{align*}
\widehat \alpha_{2} \geq \bigg(\frac{2+ \log(2m+1)}{\sinh(\beta_{\min}/2)} + \frac{1}{\alpha_{2, \min}}\bigg)^{-1},
\end{align*}
where $\beta_{\min} = \min_j \beta_j$ and $\alpha_{2, \min} = \min_j  \frac{4\sinh^2(\beta_j/4)}{\beta_j}  = \frac{4\sinh^2(\beta_{\min}/4)}{\beta_{\min}}$.
\end{theorem}

\begin{proof}
We need to show that 
\begin{align}\label{eq:multimode-LSI-bound}
\Ent_{2, \widehat \sigma}(X) \leq \bigg(\frac{2+ \log(2m+1)}{\sinh(\beta_{\min}/2)} + \frac{1}{\alpha_{2, \min}}\bigg) \big\langle X, \widehat \cL^*(X)\big\rangle_{\widehat \sigma},
\end{align}
for any $X$ satisfying $\tr\big(\Gamma_{\widehat \sigma}^{1/2}(X)^2 H_m\big)<+\infty$.
We use Lemma~\ref{lem:entropic-ineq} for the choice of $w_\ell=e^{-\frac c2 |\ell|}$ where $c>0$ and $|\ell| = \sum_{j=1}^m |\ell_j|$. We have $\log\|w\|_2^2 - \log w_\ell^2 = c|\ell| +m\log\frac{e^c+1}{e^c-1}$ and 
\begin{align}\label{eq:proof-m-LSI-1}
\Ent_{2, \widehat \sigma}(X) \leq  m\log\frac{e^c+1}{e^c-1}\|X\|_{2, \widehat \sigma}^2 + c\sum_{\ell} |\ell|\cdot\|X_\ell\|_{2, \widehat \sigma}^2  + \sum_\ell  \Ent_{2, \widehat \sigma}(I_{2,2}(X_\ell)).
\end{align}
Now, since $X_\ell\in \mathcal F_\ell$, by Corollary~\ref{cor:spectrum} we have  
$$\sum_\ell |\ell| \cdot \|X_\ell\|_{2, \widehat \sigma}^2\leq  \frac{1}{\sinh(\beta_{\min}/2)}  \sum_\ell \big\langle X_\ell, \widehat \cL^*(X_\ell) \big\rangle_{\widehat \sigma} = \frac{1}{\sinh(\beta_{\min}/2)} \big\langle X, \widehat \cL^*(X)\big\rangle_{\widehat \sigma}.$$
Next, using the fact that $X_\ell\in \mathcal F_\ell$, it is easily verified that the operator $I_{2,2}(X_\ell)$ is diagonal. On the other hand, the restriction of $\widehat \cL^*$ to diagonal operators is essentially a classical Markov semigroup for which the tensorization property holds. Thus, the $2$-log-Sobolev constant of $\widehat \cL^*$ restricted to diagonal operators equals the minimum of the $2$-log-Sobolev constants of $\cL^*_j$'s restricted to (single-mode) diagonal operators. The latter quantity is computed in Theorem~\ref{thm:p-LSI}. Putting these together we conclude that
$$\alpha_{2, \min} \Ent_{2, \widehat \sigma}(I_{2, 2}(X_\ell)) \leq \big\langle I_{2, 2}(X_\ell), \widehat \cL^*\big(I_{2, 2}(X_\ell)\big)\big\rangle_{\widehat \sigma}.$$
Using the above two inequalities in~\eqref{eq:proof-m-LSI-1} we find that 
\begin{align*}
\Ent_{2, \widehat \sigma}(X)  & \leq m\log\frac{e^c+1}{e^c-1} \|X\|_{2, \widehat \sigma}^2 + \frac{c}{\sinh(\beta_{\min}/2)}\big\langle X, \widehat \cL^*(X)\big\rangle_{\widehat \sigma} + \frac{1}{\alpha_{2, \min}}\sum_\ell  \big\langle I_2(X_\ell), \widehat \cL^*\big(I_2(X_\ell)\big)\big\rangle_{\widehat \sigma}.
\end{align*} 
It is shown in~\cite[Lemma 23]{BDR20} that for any operator $Y$, 
$$\big\langle I_{2,2}(Y),\widehat \cL^*(I_{2,2}(Y))\big\rangle_{\widehat \sigma} + \big\langle I_{2,2}(Y^\dagger), \widehat \cL^*(I_{2,2}(Y^\dagger))\big\rangle_{\widehat \sigma}\leq \big \langle Y, \widehat \cL^*(Y)\big\rangle_{\widehat \sigma}+\big\langle Y^\dagger, \widehat \cL^*(Y^\dagger)\big\rangle_{\widehat \sigma}.$$
Therefore, using $X_\ell^\dagger= X_{-\ell}$, the above inequalities imply 
\begin{align*}
\Ent_{2, \widehat \sigma}(X)  & \leq m\log\frac{e^c+1}{e^c-1} \|X\|_{2, \widehat \sigma}^2 + \frac{c}{\sinh(\beta_{\min}/2)}\langle X, \widehat \cL^*(X)\rangle_{\widehat \sigma} +  \frac{1}{\alpha_{2, \min}}\sum_\ell  \big\langle X_\ell, \widehat \cL^*(X_\ell)\big\rangle_{\widehat \sigma}\\
& =m \log\frac{e^c+1}{e^c-1} \|X\|_{2, \widehat \sigma}^2 +\Big(\frac{c}{\sinh(\beta_{\min}/2)} + \frac{1}{\alpha_{2, \min}} \Big)\big\langle X, \widehat \cL^*(X)\big\rangle_{\widehat \sigma}. 
\end{align*}
Next, by Proposition~\ref{prop:spectrum1} the generator $\widehat\cL^*$ has the spectral gap $\sinh(\beta_{\min}/2)$. Thus,  by~\cite[Theorem 4.2]{OZ99} the above inequality implies
\begin{align*}
\Ent_{2, \widehat \sigma}(X) & \leq \bigg( \frac{1+m\log\frac{e^c+1}{e^c-1} }{\sinh(\beta_{\min}/2)} +\frac{c}{\sinh(\beta_{\min}/2)} + \frac 1{\alpha_{2, \min}} \bigg)\big\langle X, \widehat \cL^*(X)\big\rangle_{\widehat \sigma}. 
\end{align*}
Finally, letting $c=\log(2m+1 )$ and using $\log(1+\frac 1m)\leq \frac 1m$ the desired inequality~\eqref{eq:multimode-LSI-bound} is obtained.

\end{proof}

%************************************************************************
\section{Proof of the CMOE conjecture}\label{sec:CMOE}

In this section, we present an alternative proof of the CMOE conjecture for single-mode phase-covariant Gaussian channels, which is first established in~\cite{DePalma+2017PRL, PalmaTrevisanGiovannettt2017Attenuator}.

\begin{theorem}\label{thm:CMOE} \emph{[CMOE Conjecture]}
For any single-mode phase-covariant Gaussian channel $\Phi$, and any quantum state $\rho$ with finite mean photon number we have
\begin{align}\label{eq:CMOE}
S\big(\Phi(\rho)\big) \geq S\big(\Phi(\tau)\big),
\end{align}
where $\tau$ is a single-mode thermal state satisfying $S(\rho) = S(\tau)$.
\end{theorem}

To prove this theorem we first state a consequence of Theorem~\ref{thm:main-technical-result}.

\begin{theorem}\label{thm:entropy-derivative-optimization}
Let $\{\Phi_t:\, t\geq 0\}$ denote a semigroup of single-mode phase-covariant Gaussian channels. Then, for any $\alpha> 0$ we have
\begin{align}\label{eq:entropy-derivative-optimization}
\eta_\thermal(\alpha)=\inf_{\rho}\, \frac{\dd}{\dd t} S\big(\Phi_t(\rho) \big)\Big|_{t=0} + \alpha S(\rho),
\end{align}
where the infimum is taken over all quantum states $\rho$ with finite mean photon number, and $\eta_\thermal(\alpha)$ is the infimum of the same function restricted to thermal states. 
\end{theorem}

\begin{proof}[Proof of Theorem~\ref{thm:entropy-derivative-optimization}]
Let $\cL$ be the Lindbladian of the semigroup $\{\Phi_t:\, t\geq 0\}$ so that $\Phi_t=e^{-t\cL}$. We have
\begin{align}\label{eq:entropy-derivative-1}
\frac{\dd}{\dd t} S\big(\Phi_t(\rho) \big)\Big|_{t=0} =  \langle \cL(\rho), \log \rho\rangle.
\end{align}
Then, the claim follows from Theorem~\ref{thm:main-technical-result} for $p=1$ and appropriate choices of $\nu_0, \nu_1\geq 0$ and $\omega=0$.

\end{proof}

We need yet another technical ingredient to prove Theorem~\ref{thm:CMOE}. For $0<x<1$ let 
\begin{align}\label{eq:tau-x-1}
\tau_x = (1-x)\sum_{n} x^n\ketbra{n}{n},
\end{align}
be a thermal state with parameter $x$. Then, letting $\Phi_t= e^{-t\cL}$ with $\cL=\nu_0\cL_0+\nu_1\cL_1$ and using~\eqref{eq:entropy-derivative-1}, we find that the objective function in~\eqref{eq:entropy-derivative-optimization} for thermal states equals
\begin{align*}
f_\alpha(x) &:=  \frac{\dd}{\dd t} S\big(\Phi_t(\tau_x) \big)\big|_{t=0} + \alpha S(\tau_x) \\
&\, = \frac{1}{1-x}\Big(\!  (\nu_1 x-\nu_0) \log x - \alpha x\log x-\alpha(1-x)\log(1-x)\! \Big).
\end{align*}
This equation can also be derived by taking the limit of $p\to1^+$ in~\eqref{eq:Upsilon-for-thermal}.

Assume that $\nu_0>0$ and as before let $\alpha>0$. Then, we have $\lim_{x\to 0^+ }f_\alpha(x) = \lim_{x\to 1^- }f_\alpha(x) = +\infty$. This means that the minimum of $f_\alpha(x)$ is attained at roots of $f'_\alpha(x)$. We have
\begin{align*}
f'_\alpha(x) &= \frac{1}{1-x}\Big( \nu_1\log x + \frac{\nu_1x-\nu_0}{x} - \alpha\log x + \alpha\log(1-x)  \Big) \\
&\quad~ +\frac{1}{(1-x)^2} \Big((\nu_1x-\nu_0)\log x -\alpha x\log x - \alpha(1-x)\log(1-x) \Big)\\
& =\frac{1}{(1-x)^2}\bigg(\! (\nu_1-\nu_0)\log x + \frac{(1-x)(\nu_1x-\nu_0)}{x} -\alpha\log x\bigg)
\end{align*}
Therefore, $f'_\alpha(x)=0$ iff $g(x)=\alpha$ where
$$g(x) := \nu_1-\nu_0 + \frac{(1-x)(\nu_1x-\nu_0)}{x\log x}.$$
Computing the derivative of $g(x)$, we have
$$g'(x) = \frac{\nu_0(\log x-x+1) + \nu_1(x(x-1)-x^2\log x)}{x^2\log^2x}.$$
Then, using $\log x\leq x-1$ and $-\log x= \log\frac 1x\leq \frac 1x-1$ we find that $g'(x)\leq 0$ for $0<x<1$, and that $g(x)$ is monotone decreasing in this interval. Moreover, since $\nu_0>0$ we have $\lim_{x\to 0^+} g(x)=+\infty$ and $\lim_{x\to 1^-} g(x) = 0$. This means that $g(x)\geq 0$ for all $0<x<1$, and $g(x)$ takes any positive value in this interval exactly once. 
As a conclusion, $g(x)=\alpha$ has exactly one solution in $(0,1)$ for any $\alpha>0$. We conclude that, if $\alpha>0$, the minimum of $f_\alpha(x)$ is attained at the unique root of $f'_\alpha(x)$. More importantly, for any $0<x_0<1$, there is $\alpha_{0}=g(x_0)$ such that the minimum of $f_{\alpha_{0}}(x)$ is attained at $x=x_0$.

\begin{proof}[Proof of Theorem~\ref{thm:CMOE}]
As a phase-covariant Gaussian channel, we have $\Phi=\Phi_{t_1}$ for some phase-covariant Gaussian semigroup $\{\Phi_t:\, t\geq 0\}$ and some $t_1>0$.\footnote{If $t_1=0$, then $\Phi$ is the identity channel and there is nothing to prove.} Let $\cL=\nu_0\cL_0+\nu_1\cL_1$ be the Lindbladian of this semigroup. In order to apply the above computations, by taking the limit of $\nu_0\to 0^+$ we assume with no loss of generality that $\nu_0>0$. The point is that if $S\big(\Phi_{t_1}(\rho)\big)\geq S(\Phi_{t_1}(\tau))$ holds for any $\nu_0>0$, then by the continuity of von Neumann entropy~\cite[Lemma 18]{winter2016tight}, it holds for $\nu_0=0$ as well.

Suppose that $S\big(\Phi_{t_1}(\rho)\big)< S(\Phi_{t_1}(\tau_{x_0}))$ where $\tau=\tau_{x_0}$ is the thermal state in the statement of the theorem with parameter $x_0$ as in~\eqref{eq:tau-x-1}. We note that by assumption $S(\Phi_{t}(\rho))= S(\Phi_{t}(\tau_{x_0}))$ for $t=0$. Thus, we may define
$$t_0 =\sup\big\{t:\,   S\big(\Phi_{t}^{\otimes m}(\rho)\big)\geq mS\big(\Phi_{t}(\tau_{x_0})\big),   t<t_1  \big\}.$$
Then, by continuity we have $S(\Phi_{t_0}(\rho))= S(\Phi_{t_0}(\tau_{x_0}))$ and 
\begin{align}\label{eq:S-t-t1}
S(\Phi_{t}(\rho))< S(\Phi_{t}(\tau_{x_0})),\qquad \forall t_0<t\leq t_1.
\end{align}

Recall that $\Phi_t(\tau_{x_0})$ is a thermal state for any $t>0$. Let $x_t$ be the parameter of this thermal state: $\tau_{x_t} = \Phi_t(\tau_{x_0})$. We note that $t\mapsto x_t$ is smooth. Then, by the above discussions, there is a smooth function $t\mapsto \alpha_t$ with $\alpha_t>0$ such that 
$$\eta_\thermal(\alpha_t)= \frac{\dd}{\dd s} S\big(\Phi_s(\tau_{x_t}) \big)\big|_{s=0} + \alpha_t S(\tau_{x_t}).$$
Employing Theorem~\ref{thm:entropy-derivative-optimization}, for any $t_0\leq t\leq t_1$ we obtain
$$ \frac{\dd}{\dd t} S\big(\Phi_t(\rho) \big) + \alpha_t S(  \Phi_t(\rho))   \geq \frac{\dd}{\dd t} S(\tau_{x_t} ) + \alpha_t S(\tau_{x_t}).$$
Taking the integral of both sides yields
\begin{align*}
   S\big(\Phi_{t_1}(\rho) \big) - S\big(\Phi_{t_0}(\rho) \big) + \int_{t_0}^{t_1}\!\alpha_t S(  \Phi_t(\rho))\,\dd t     \geq S(\tau_{x_{t_1}}) - S(\tau_{x_{t_0}}) +\int_{t_0}^{t_1}\! \alpha_t S(\tau_{x_t})\, \dd t.
\end{align*}
We note that $S(\Phi_{t_0}(\rho))= S(\tau_{x_{t_0}})$, so this inequality is in contradiction with~\eqref{eq:S-t-t1}. This implies that the starting inequality $S\big(\Phi_{t_1}(\rho)\big)< S(\Phi_{t_1}(\tau_{x_0}))$ does not hold. 

\end{proof}

Generalization of the CMOE conjecture to the multimode case is an open problem in general. However, generalizing the proof idea of Theorem~\ref{thm:CMOE} and using Theorem~\ref{thm:main-technical-result-multimode-Gaussian}, we can prove the conjecture for a special class of multimode states that can be prepared by applying a passive Gaussian unitary $U_G$ on an arbitrary product state, $\rho=U_G \rho_{1}\otimes\rho_{2}\otimes \cdots\otimes\rho_{m} U_G^\dagger$.

%************************************************************************
\section{Conclusion}\label{sec:con}

In this paper, we introduced a meta log-Sobolev inequality, formalized by Theorem~\ref{thm:main-technical-result}, that provides a general information theoretic framework to study phase-covariant Gaussian quantum channels. In general, this inequality and our technical proof, particularly the tensorization argument that we used in the second step, are of independent interest. However, by using this result, we have derived new optimal bounds in the context of $p$-log-Sobolev inequalities for the semigroup of attenuation channels.

Specifically, in Theorem~\ref{thm:p-LSI}, we explicitly computed the optimal $p$-log-Sobolev constant of the quantum Ornstein-Uhlenbeck semigroup for any $1\leq p\leq 2$. We also in Theorem~\ref{thm:multimode-LSI} established a bound on the $2$-log-Sobolev constant in the multimode case that scales logarithmically with the number of modes. We conjecture that the log-Sobolev constants in the multimode case match those of the single-mode case and the tensorization property holds in general. We showed in Corollary~\ref{cor:LSI-Gaussian} that this conjecture holds assuming that the optimal multimode states are Gaussian.

Quantum log-Sobolev inequalities are also defined for $0\leq p< 1$ and are related to \emph{reverse} hypercontractivity inequalities~\cite{CKMT15, BDR20}. An interesting open problem is to compute the $p$-log-Sobolev constant for this range of parameter, which we leave it for future works.

The proof of our multimode log-Sobolev inequality is based on an entropic inequality in Lemma~\ref{lem:entropic-ineq} which is of independent interest. Using this inequality, one might be able to prove log-Sobolev inequalities for the tensor product of other quantum Markov semigroups. We leave exploration of such applications of Lemma~\ref{lem:entropic-ineq} for future works.

Our results have novel consequences even in the fully classical case, i.e., when restricting to states $\rho$ that are diagonal in the Fock basis. In particular, as pointed out in~\cite{CarboneSasso2008}, restricting the quantum Ornstein-Uhlenbeck semigroup to diagonal states, we obtain a classical \emph{birth-death process} whose log-Sobolev constants were not known prior to our work (except in the case of $p=1$). Moreover, our results have consequences for the classical operation of \emph{thinning} on the space of integer-valued random variables, for which we refer to~\cite{DePalma+2018survey} and references therein.

We also presented an alternative proof of the CMOE conjecture in Theorem~\ref{thm:CMOE} in the single-mode case. This conjecture in the multimode case is left as a major open problem. 
Our results contribute to the resolution of \emph{Gaussian optimizer conjectures} in the quantum case~\cite{DePalma+2018survey}. These conjectures state the optimality of Gaussian states for certain optimization problems regarding Gaussian channels. Theorem~\ref{thm:main-technical-result} establishes this conjecture for the quite general function $\Upsilon(\rho)$ in the single-mode case. Generalizing this theorem for multimode states can resolve the CMOE conjecture in the multimode case. To this end, ideas developed in~\cite{GaoRouze2022Complete} might be useful. We leave further exploration of this problem for future works.

\medskip
\paragraph{Acknowledgements.} SB is thankful to Reza Seyyedali for several inspiring discussions. This work is supported by the NRF grant NRF2021-QEP2-02-P05 and the Ministry of Education, Singapore, under the Research Centres of Excellence program.

%****************************************REFERENCES***********************************************
{\small
\bibliographystyle{abbrvurl} 
\bibliography{References}
}

%********************************************************************************************************
%*****************************************APPENDIX**************************************************

\appendix

\section{Generators of phase-covariant Gaussian channels}\label{app:generators}

In this appendix we show that any single-mode phase-covariant channel given by~\eqref{eq:Phas-Cov-Def} is a member of a semigroup whose Lindbladian equals~\eqref{eq:LG-general form}.  

Let us consider the family of channels $\{\Phi_t:\, t\geq 0\}$ corresponding to parameters $\{(\lambda_t, \gamma_t):\, t\geq 0\}$ in~\eqref{eq:Phas-Cov-Def}:
\begin{equation}\label{eq:in-out-charac}
	\chi_{\Phi_t(\rho)}(\xi)=\exp\!\bigg(\!\!-\frac 1 2 \gamma_t|\xi|^2\bigg)\chi_\rho\big(\!\sqrt{\lambda_t}\,\xi\big).
\end{equation}
Suppose that $\{\Phi_t:\, t\geq 0\}$ is a semigroup which in particular means that $\Phi_0$ is the identity channel corresponding to parameters $\gamma_0=0$ and $\lambda_0=1$. Then, the generator of this semigroup is determined by
\begin{equation*}
		\chi_{\cL(\rho)}(\xi)=-\frac{\dd}{\dd t} \chi_{\Phi_{t}(\rho)}(\xi)\Big|_{t=0} .
\end{equation*}
Using \eqref{eq:in-out-charac} we compute the derivative 
\begin{align*}
\frac{\dd}{\dd t} \chi_{\Phi_{t}(\rho)}(\xi)\Big|_{t=0} &= - \frac{1}{2}\gamma'_0\,|\xi|^2 \tr(\rho D_\xi)+\frac{1}{2}\lambda'_0\,\Big(\xi\, \tr\big(\rho\, \bfa^\dagger D_\xi\big) - \bar\xi\, \tr\big(\rho\, \bfa D_\xi\big)\Big)\\
&=- \frac{1}{2}\gamma_0'\, \tr\big(\rho  \big[\bfa, [\bfa^\dagger ,D_\xi]\big]\big)+\frac{1}{2}\lambda_0' \Big(\tr\big(\rho\, \bfa^\dagger \big[\bfa, D_\xi\big]\big) -  \tr\big(\rho\, \bfa \big[\bfa^\dagger, D_\xi\big]\big) \Big)\\
&=- \frac{1}{2}\gamma_0'\, \tr\big(\big[\bfa^\dagger, [\bfa ,\rho]\big]  D_\xi\big) + \frac{1}{2}\lambda'_0 \Big(\tr\big(\big[\bfa^\dagger, \rho\, \bfa\big]D_\xi -\big[\bfa, \rho\, \bfa^\dagger\big]D_\xi\big)\Big)\\
&=- \frac{1}{2}\gamma'_0\, \chi_{[\bfa^\dagger, [\bfa ,\rho]]}(\xi) + \frac{1}{2}\lambda'_0\,\chi_{[\bfa^\dagger, \rho \bfa] -[\bfa, \rho \bfa^\dagger]}(\xi),
\end{align*}
where $\lambda'_0=\frac{\dd\lambda_t}{\dd t}\big|_{t=0}$, $\gamma'_0=\frac{\dd\gamma_t}{\dd t}\big|_{t=0}$. Here, we use   $\frac{\dd}{\dd \lambda} D_{\lambda\delta} =  (\xi \bfa^\dagger - \bar\xi \bfa) D_{\lambda\xi}$ in the first line, and the well-known equations
$[\bfa, D_\xi] =\xi D_\xi$ and $[\bfa^\dagger, D_\xi] =\bar \xi D_\xi$ in the second line. Replacing the characteristic functions in the above equation by their corresponding operators yields
\begin{align*}
	\frac{\dd}{\dd t} \Phi_{t}(\rho)\Big|_{t=0}=&- \frac{1}{2}\gamma'_0\, \big[\bfa^\dagger, [\bfa ,\rho]\big]+ \frac{1}{2}\lambda'_0\Big([\bfa^\dagger, \rho\, \bfa] -[\bfa, \rho\, \bfa^\dagger]\Big)\\
	=&-\frac{1}{2}\gamma_0' \Big(\frac 12\big\{\bfa^\dagger \bfa, \rho\big\}+\frac 12\big\{\bfa \bfa^\dagger, \rho\big\}- \bfa^\dagger\rho\,\bfa-\bfa\,\rho\,\bfa^\dagger\Big)\\
	&\, +\frac{1}{2}\lambda'_0\Big(\frac 12\big\{\bfa^\dagger \bfa, \rho\big\}-\frac 12\big\{\bfa \bfa^\dagger, \rho\big\}+ \bfa^\dagger\rho\,\bfa-\bfa\,\rho\,\bfa^\dagger\Big)\\
	=&-\frac{1}{2} \big(\gamma'_0+\lambda'_0\big)\Big(\frac 12\big\{\bfa \bfa^\dagger, \rho\big\}-\bfa^\dagger\rho\bfa\Big)-\frac{1}{2} \big(\gamma'_0-\lambda'_0\big)\Big(\frac 12\big\{\bfa^\dagger \bfa, \rho\big\}-\bfa\rho\bfa^\dagger\Big)\\
	=&-\frac{1}{2} \big(\gamma'_0+\lambda'_0\big)\cL_{0}(\rho)-\frac{1}{2} \big(\gamma'_0-\lambda'_0\big)\cL_1(\rho),
\end{align*}  
where we define 
$$\cL_{0}(\rho)=\frac 12\big\{\bfa \bfa^\dagger, \rho\big\}-\bfa^\dagger\rho\bfa \qquad\text{and}\qquad \cL_1(\rho)=\frac 12\big\{\bfa^\dagger \bfa, \rho\big\}-\bfa\rho\bfa^\dagger.$$ 
Therefore, if the channels defined by~\eqref{eq:in-out-charac} form a semigroup, then the corresponding Lindbladian is equal to 
\begin{equation}\label{eq:gen-gen}
	\cL(\rho)=\frac{1}{2} \big(\gamma_0'+\lambda_0'\big)\cL_{0}(\rho)+\frac{1}{2} \big(\gamma_0'-\lambda_0'\big)\cL_1(\rho).
\end{equation}
We use this equation to obtain the generator of the three classes of phase-covariant Gaussian channels.

\subsection{Attenuator channels}
The attenuator channel can be physically modeled by overlapping the input state and a thermal state $\sigma_\beta$ on a beam splitter with transmissivity $\lambda$:
\begin{equation*}
	\Phiatt_\lambda(\rho)=\tr_2\Big( U_{\text{BS},\lambda} (\rho\otimes \sigma_\beta) U_{\text{BS},\lambda}^\dagger \Big),
\end{equation*}
where the 2-mode unitary operator $U_{\text{BS},\lambda}$ can be described by
\begin{align}\label{eq:U-almbda-a1-a2}
	\begin{cases}
		U_{BS,\lambda}^\dagger \bfa_1 U_{BS,\lambda} &= \sqrt {\lambda}\, \bfa_1 + \sqrt{1-{\lambda}}\, \bfa_2,\\ 
		U_{BS,\lambda}^\dagger \bfa_2 U_{BS,\lambda}  &= -\sqrt{1-{\lambda}}\, \bfa_1 + \sqrt{\lambda}\, \bfa_2.
	\end{cases}
\end{align}
Using these relations, the characteristic function of the two-mode state $U_{\text{BS},\lambda} (\rho\otimes \sigma_\beta) U_{\text{BS},\lambda}^\dagger$ equals
\begin{align*}
\tr\Big( U_{\text{BS},\lambda} (\rho\otimes \sigma_\beta) U_{\text{BS},\lambda}^\dagger D_{\xi_1}\otimes D_{\xi_2}\Big)&=\tr\Big( \rho\otimes \sigma_\beta U_{\text{BS},\lambda}^\dagger (D_{\xi_1}\otimes D_{\xi_2})U_{\text{BS},\lambda}\Big)\\
&=\tr\Big(\! \rho D_{\sqrt {\lambda} \xi_1 -\sqrt{1-{\lambda}}\xi_2}\Big)\,\tr\Big(\! \sigma_\beta D_{\sqrt{1-{\lambda}}\xi_1 +\sqrt {\lambda} \xi_2}\Big).
\end{align*}
Thus, the characteristic function of the output state of the attenuation channel is given by
\begin{align*}
\chi_{\Phiatt_\lambda(\rho)}(\xi)& =\chi_\rho\big(\sqrt {\lambda}\,\xi_1 -\sqrt{1-{\lambda}}\,\xi_2\big)\,\chi_{\sigma_\beta}\big(\sqrt{1-{\lambda}}\,\xi_1 +\sqrt {\lambda} \,\xi_2\big)\Big|_{(\xi_1, \xi_2)=(\xi, 0)}\\
&= \chi_{\rho}\big(\sqrt {\lambda}\, \xi\big)\, \chi_{\sigma_\beta}\big(\sqrt{1-{\lambda}}\,\xi\big)\\
&= \exp\!\Big({\!{-\frac{1}{2} \coth(\beta/2) (1-\lambda) |\xi|^2}}\Big)  \chi_{\rho}\big(\sqrt {\lambda}\, \xi\big),
\end{align*}
where we replace the characteristic function of the thermal state from~\eqref{eq:chi-thermal}.

Setting $\lambda=\lambda_t=e^{-2ct}$ with $c>0$ in the above relation, we observe that attenuator channels $\{\Phiatt_{\lambda_t}:\, t\geq0\}$ form a semigroup. Comparing the above equation with \eqref{eq:in-out-charac} also gives $\gamma_t=\coth(\beta/2) (1-\lambda_t)$. Therefore, by plugging $\lambda_0'=-2c$ and $\gamma_0'=2c\coth(\beta/2) $ into \eqref{eq:gen-gen}, we obtain the corresponding generator
\begin{align*}
	\cL(\rho)&=c \big(\!\coth(\beta/2)-1\big)\cL_{0}(\rho)+c \big(\!\coth(\beta/2)+1\big)\cL_1(\rho).
\end{align*} 

%%%%%%%%%%%%%%%%%%%%%%%%%%%%%%%%%%%%%%%%%%%%%%%%%%%%%%
\subsection{Amplifier channels}
Following a similar procedure as for attenuator channels, we can derive the generator of the semigroup of amplifier channels. An amplifier channel can be modeled by applying a two-mode squeezing unitary with squeezing parameter $\lambda\geq1$ on the input state and a thermal state $\sigma_\beta$, and then tracing over the second mode:
\begin{equation*}
	\Phiamp_\lambda(\rho)=\tr_2\Big( U_{\text{2S},\lambda} (\rho\otimes \sigma_\beta) U_{\text{2S},\lambda}^\dagger \Big).
\end{equation*}
Here, the action of the two-mode squeezing unitary can be described by
\begin{align}\label{eq:U2S-almbda-a1-a2}
	\begin{cases}
		U_{2S,\lambda}^\dagger \bfa_1 U_{2S,\lambda} &= \sqrt {\lambda}\, \bfa_1 + \sqrt{{\lambda}-2}\, \bfa_2^\dagger,\\ 
		U_{2S,\lambda}^\dagger \bfa_2 U_{2S,\lambda}  &= \sqrt{\lambda}\, \bfa_2+\sqrt{{\lambda}-1}\, \bfa_1^\dagger.
	\end{cases}
\end{align}
Using these relations, we have
\begin{equation*}
U_{2S,\lambda}^\dagger D_{\xi_1}\otimes D_{\xi_2} U_{2S,\lambda}=D_{\xi_1\sqrt{\lambda}-\bar{\xi}_2\sqrt{\lambda-1}}\otimes D_{\xi_2\sqrt{\lambda}-\bar{\xi}_1\sqrt{\lambda-1}}.
\end{equation*}
Thus, the characteristic function of the output state is given by
\begin{align*}
	\chi_{\Phiamp_\lambda(\rho)}(\xi)&=\chi_\rho\big(\sqrt{\lambda}\,\xi_1-\sqrt{\lambda-1}\,\bar{\xi}_2\big)\,\chi_{\sigma_\beta}\big(\sqrt{\lambda}\,\xi_2-\sqrt{\lambda-1}\,\bar{\xi}_1\big)\Big|_{(\xi_1, \xi_2)=(\xi, 0)}\\
	&= \chi_{\rho}\big(\sqrt {\lambda}\,\xi\big)\, \chi_{\sigma_\beta}\big(\sqrt{{\lambda}-1}\,\xi\big)\\
	&= \exp\!\Big(\!{{-\frac{1}{2} \coth(\beta/2) (\lambda-1) |\xi|^2}}\Big)  \chi_{\rho}\big(\sqrt {\lambda}\,\xi\big),
\end{align*}
where in the last line the characteristic function of the thermal state \eqref{eq:chi-thermal} is replaced.

Setting $\lambda=\lambda_t=e^{2ct}$ with $c>0$, we observe that amplifier channels $\{\Phiamp_{\lambda_t}:\, t\geq0\}$ form a semigroup. Here, $\gamma_t=\coth(\beta/2) (\lambda_t-1)$. Thus, by plugging $\lambda_0'=2c$ and $\gamma_0'=2c\coth(\beta/2)$ in \eqref{eq:gen-gen}, we obtain the generator
\begin{align*}
	\cL(\rho)&=c \big(\!\coth(\beta/2)+1\big)\cL_0(\rho)+c \big(\!\coth(\beta/2)-1\big)\cL_{1}(\rho).
\end{align*} 
Note that here we get the same coefficients as those of the generator of the attenuator semigroup, but swapped.

%%%%%%%%%%%%%%%%%%%%%%%%%%%%%%%%%%%%%%%%%%%%%%
\subsection{Additive-noise channels}  
The classical additive-noise Gaussian channels can be modeled by applying a displacement operator whose parameter is chosen at random according to a Gaussian probability distribution:
\begin{equation*}
	\Phiad_\gamma(\rho)=\frac{2}{\pi\gamma}\int e^{-\frac{2}{\gamma}|\alpha|^2} D_\alpha\rho D^{\dagger}_\alpha\, \dd^2\alpha.
\end{equation*} 
The characteristic function of the output state of this channel is given by
\begin{align*}
\chi_{\Phiad_\gamma(\rho)}(\xi)&=\frac{2}{\pi\gamma}\int e^{-\frac{2}{\gamma}|\alpha|^2} \tr\big(D_\alpha\rho D^{\dagger}_\alpha D_\xi\big)\, \dd^2\alpha\\
&=\frac{2}{\pi\gamma}\chi_\rho(\xi) \int e^{-\frac{2}{\gamma} |\alpha|^2} e^{\xi\bar{\alpha}-\bar{\xi}\alpha}\, \dd^2\alpha\\
&=\exp\!\Big({-\frac 12\gamma|\xi|^2}\Big)\chi_\rho(\xi),
\end{align*}
where we use $D^{\dagger}_\alpha D_\xi D_\alpha=\exp(\xi\bar{\alpha}-\bar{\xi}\alpha) D_\xi$ in the second line. 

By choosing $\gamma=2ct$ with $c>0$, the channels $\{\Phiad_t:\, t\geq0\}$ form a semigroup. Here, $\lambda_t=1$. Thus, by inserting $\lambda_0'=0$ and $\gamma_0'=2c$ in~\eqref{eq:gen-gen} we obtain the generator
\begin{equation*}
\cL(\rho)=c\cL_0(\rho)+c\cL_{1}(\rho).
\end{equation*} 

%************************************************************************
\section{$\Upsilon(\rho)$ is well-defined if $\tr(\rho\bfa^\dagger \bfa)<+\infty$}\label{app:range-Upsilon}

In this appendix we show that $\Upsilon(\rho)$ given by
\begin{align*}
\Upsilon(\rho) =\hatp  \bigg(\!\nu_0 \Big[\tr\big(\rho\bfa\bfa^\dagger\big) -  \tr\big(  \rho^{1/p} \bfa \rho^{1/\hatp} \bfa^\dagger\big) \Big] +\nu_1\Big[  \tr\big(\rho\bfa^\dagger\bfa\big)-  \tr\big(  \rho^{1/p} \bfa^\dagger \rho^{1/\hatp} \bfa \big)\!\Big]\!\bigg) + \omega \tr\big(\rho \bfa^\dagger\bfa\big) + S(\rho).
\end{align*}
is well-defined for any state $\rho$ with finite mean photon number. For such a state both $\tr(\rho \bfa^\dagger\bfa), \tr(\rho \bfa \bfa^\dagger)$ are finite. Moreover, if $\rho = \sum_j \lambda_j \ketbra{\psi_j}{\psi_j}$ is the eigen-decomposition of $\rho$, then we have
\begin{align*}
 \tr\big(  \rho^{1/p} \bfa \rho^{1/\hatp} \bfa^\dagger\big) &  =  \sum_{i, j}   \lambda_i^{1/p}\lambda_j^{1/\hatp} \big| \bra{\psi_i} \bfa\ket{\psi_j}  \big|^2\\
 & \leq \sum_{i, j}  \Big( \frac 1p \lambda_i  + \frac{1}{\hatp} \lambda_j \Big) \big| \bra{\psi_i} \bfa\ket{\psi_j}  \big|^2\\
 & = \sum_i \frac{1}{p}\lambda_i \big\|\bfa^\dagger \ket{\psi_i}\big\|^2 + \sum_j \frac{1}{\hatp} \lambda_j \big\|\bfa \ket{\psi_j}\big\|^2\\
 & = \frac{1}{p}\tr\big(\rho \bfa\bfa^\dagger\big) + \frac{1}{\hatp}\tr\big(\rho  \bfa^\dagger \bfa\big),
\end{align*}
where in the second line we use Young's inequality. Thus, $ \tr\big(  \rho^{1/p} \bfa \rho^{1/\hatp} \bfa^\dagger\big)$ and similarly $\tr\big(  \rho^{1/p} \bfa^\dagger \rho^{1/\hatp} \bfa \big)$ are finite. Also, using the non-negativity of Umegaki's quantum relative entropy $D(\rho\| \tau) = \tr(\rho\log \rho) - \tr(\rho\log \tau)$, for the thermal state $\tau = (1-1/e) e^{-\bfa^\dagger \bfa} $ we have
\begin{align*}
0\leq D(\rho \| \tau) = -S(\rho) - \tr(\rho \log \tau)  = -S(\rho) -\log(1-1/e) + \tr\big(\rho \bfa^\dagger \bfa \big).
\end{align*}
Therefore, since $\tr\big(\rho \bfa^\dagger\bfa\big)$ is finite, $S(\rho)$ is also finite.

%**************************************************************************
\section{Verification of~\eqref{eq:phi-derivative}}\label{app:comp-derivative}

Recall that 
\begin{align*}
\phi(x, y) = & -y^2\log y^2 -(1-y^2)\log(1-y^2) + y^2\log x^2 + (1-y^2)\log(1-x^2) \\
&\, -\frac{\big( x^{\frac 2p}-y^{\frac2 p} \big)\big(  x^{\frac 2\hatp}-y^{\frac2 {\hatp} }\big)\log x^2}{\big(1-x^{\frac 2p}\big)\big(1-x^{\frac 2\hatp}\big)} .
\end{align*}
Then, computing the derivative term-by-term yields
\begin{align*}
\frac{\dd}{\dd x} \phi(x, y) = &\frac{2y^2}{x} - \frac{2x(1-y^2)}{1-x^2} - \frac{2\big(x^{\frac 2p}-y^{\frac 2p}\big)\big(x^{\frac 2\hatp}-y^{\frac 2\hatp}\big)}{x\big(1-x^{\frac 2p}\big)\big(1-x^{\frac 2\hatp}\big)}\\
& \, - \frac{2x^{\frac 2p}\big(x^{\frac 2\hatp} - y^{\frac 2\hatp}\big)\log x^2}{px\big(1-x^{\frac 2p}\big)\big(1-x^{\frac 2\hatp}\big)}  - \frac{2x^{\frac 2\hatp}\big(x^{\frac 2p} - y^{\frac 2p}\big)\log x^2}{\hatp x\big(1-x^{\frac 2p}\big)\big(1-x^{\frac 2\hatp}\big)} \\
&\, - \frac{2x^{\frac 2p}\big(x^{\frac 2p}-y^{\frac 2p}\big)\big(x^{\frac 1\hatp} - y^{\frac 1\hatp}\big)\log x^2}{px\big(1-x^{\frac 2p}\big)^2\big(1-x^{\frac 2\hatp}\big)}  - \frac{2x^{\frac 2\hatp}\big(x^{\frac 2\hatp}-y^{\frac 2\hatp}\big)\big(x^{\frac 2p} - y^{\frac 2p}\big)\log x^2}{\hatp x\big(1-x^{\frac 2p}\big)\big(1-x^{\frac 2\hatp}\big)^2}\\
= &   -\frac{2(x^2-y^2)}{x(1-x^2)} - \frac{2\big(x^{\frac 2p}-y^{\frac 2p}\big)\big(x^{\frac 2\hatp}-y^{\frac 2\hatp}\big)}{x\big(1-x^{\frac 2p}\big)\big(1-x^{\frac 2\hatp}\big)} \\
& \,- \frac{2x^{\frac 2p}\big(1-y^{\frac 2p}\big)\big(x^{\frac 2\hatp} - y^{\frac 2\hatp}\big)\log x^2}{px\big(1-x^{\frac 2p}\big)^2\big(1-x^{\frac 2\hatp}\big)}  - \frac{2x^{\frac 2\hatp}\big(1-y^{\frac 2\hatp}\big)\big(x^{\frac 2p} - y^{\frac 2p}\big)\log x^2}{\hatp x\big(1-x^{\frac 2p}\big)\big(1-x^{\frac 2\hatp}\big)^2}.
\end{align*}
Therefore, we have
\begin{align*}
\frac x2 \frac{\dd}{\dd x} &\phi(x, y) + \frac{x^{\frac 2p}\big(1-y^{\frac 2p}\big)\big(x^{\frac 2\hatp} - y^{\frac 2\hatp}\big)}{\big(1-x^{\frac 2p}\big)\big(1-x^{\frac 2\hatp}\big)}  \bigg[  \frac{\log x^2}{p(1-x^{\frac 2p})}+1 \bigg]  + \frac{x^{\frac 2\hatp}\big(1-y^{\frac 2\hatp}\big)\big(x^{\frac 2p} - y^{\frac 2p}\big)}{\big(1-x^{\frac 2p}\big)\big(1-x^{\frac 2\hatp}\big)} \bigg[  \frac{\log x^2}{\hatp(1-x^{\frac 2\hatp})}+1 \bigg] \\
= &   -\frac{(x^2-y^2)}{(1-x^2)} - \frac{\big(x^{\frac 2p}-y^{\frac 2p}\big)\big(x^{\frac 2\hatp}-y^{\frac 2\hatp}\big)-x^{\frac 2p}\big(1-y^{\frac 2p}\big)\big(x^{\frac 2\hatp} - y^{\frac 2\hatp}\big)-x^{\frac 2\hatp}\big(1-y^{\frac 2\hatp}\big)\big(x^{\frac 2p} - y^{\frac 2p}\big)}{\big(1-x^{\frac 2p}\big)\big(1-x^{\frac 2\hatp}\big)} \\
= &   -\frac{(x^2-y^2)}{(1-x^2)} + \frac{x^2-y^2 +x^{\frac 2p}y^2 - x^2y^{\frac 2p}+ x^{\frac 2\hatp}y^2-x^2y^{\frac 2\hatp}}{\big(1-x^{\frac 2p}\big)\big(1-x^{\frac 2\hatp}\big)} \\
= &    \frac{-(x^2-y^2)\big(1-x^{\frac 2p}\big)\big(1-x^{\frac 2\hatp}\big) +\big(1-x^2\big)\big(x^2-y^2 +x^{\frac 2p}y^2 - x^2y^{\frac 2p}+ x^{\frac 2\hatp}y^2-x^2y^{\frac 2\hatp}\big)}{\big(1-x^2\big)\big(1-x^{\frac 2p}\big)\big(1-x^{\frac 2\hatp}\big)} \\
= &    \frac{x^2\Big( \! \big(1- x^{\frac 2\hatp}\big)\big(1-y^{\frac 2\hatp}\big) \big(x^{\frac 2p} - y^{\frac 2p}\big)   + \big(1- x^{\frac 2p}\big)\big(1-y^{\frac 2p}\big) \big(x^{\frac 2\hatp} - y^{\frac 2\hatp}\big) \!  \Big)
}{\big(1-x^2\big)\big(1-x^{\frac 2p}\big)\big(1-x^{\frac 2\hatp}\big)}.
\end{align*}
This gives~\eqref{eq:phi-derivative}.

%***************************************************
\section{Proof of $\alpha_2=\min_{p\geq 1} \alpha_p$}\label{app:alpha2-min}

In this appendix, we show that for any $p\geq 1$ the following inequality holds:
\begin{align*}
\frac{p\hatp}{4}\big(1-e^{-\beta/p}\big)\big(1-e^{-\beta/\hatp}\big)\geq  \big(1-e^{-\beta/2}\big)^2.
\end{align*}
First, we may restrict to $p\in [1, 2]$. Then, with the change of variable $p=\frac{2}{1+t}$ with $t\in [0,1]$ the above inequality is equivalent to
\begin{align*}
\big(1-e^{-\frac \beta 2 (1+t) }\big) \big(1-e^{-\frac \beta 2 (1-t) }\big)\geq (1-t^2) \big(1-e^{-\beta/2}\big)^2.
\end{align*}
Starting with the left hand side we compute
\begin{align*}
\big(1-e^{-\frac \beta 2 (1+t) }\big) \big(1-e^{-\frac \beta 2 (1-t) }\big) &  = 1+e^{-\beta} - e^{-\beta/2}\big(e^{\beta t/2}+e^{-\beta t/2}\big)\\
& = 1+ e^{-\beta} -2 e^{-\beta/2} \sum_{k:\text{ even}} \frac{1}{k!} \Big(\frac{\beta t}{2}\Big)^k\\
& \geq 1+ e^{-\beta} -2 e^{-\beta/2}\bigg( 1+ t^2\sum_{k\geq 2:\text{ even}} \frac{1}{k!} \Big(\frac{\beta }{2}\Big)^k \bigg)\\
&  = \big(1-e^{-\beta/2}\big)^2 -t^2 e^{-\beta/2} \Big( e^{\beta/2} + e^{-\beta/2} -2  \Big)\\
& = (1-t^2)  \big(1-e^{-\beta/2}\big)^2.
\end{align*}
Here, the inequality holds since $0\leq t\leq 1 $. We are done.

%**************************************************************************
\section{Proof of Proposition~\ref{prop:spectrum1}}\label{app:spectrum}

It suffices to prove the proposition for $m=1$. In this case, we drop subscript $j$ and denote $\widehat \cL$ and $\widehat \sigma$ simply by $\cL$ and $\sigma$, respectively.

Using the commutation relation $[\bfa , \bfa^\dagger]=1$, we find that $[\bfa, \bfq_z ] = \frac{z}{\sqrt 2}$, and
$$\big[\bfa, e^{s\bfq_z - \frac{\coth(\beta/2)} 4 s^2}\big] = \frac{z}{\sqrt 2} \Big( \frac{\dd}{\dd t} e^{st-\frac{\coth(\beta/2)} 4s^2}\Big)\Big|_{t=\bfq_z}=\frac{z}{\sqrt 2} s e^{s\bfq_z - \frac{\coth(\beta/2)} 4 s^2}.$$
This means that 
$$e^{s\bfq_z - \frac{\coth(\beta/2)} 4 s^2}\bfa = \Big(\bfa-\frac{z}{\sqrt 2} s\Big)e^{s\bfq_z - \frac{\coth(\beta/2)} 4 s^2},$$
and similarly
$$e^{s\bfq_z - \frac{\coth(\beta/2)} 4 s^2}\bfa^\dagger = \Big(\bfa^\dagger+\frac{\bar z}{\sqrt 2} s\Big)e^{s\bfq_z - \frac{\coth(\beta/2)} 4 s^2}.$$
Using these and $|z|=1$, a straightforward computation shows that 
\begin{align*}
\mathcal{L}^*(e^{s\bfq_z - \frac {\coth(\beta/2)}4 s^2})& = e^{\beta/2}\Big(  \frac 12 \bfa^\dagger\bfa +\frac 12(\bfa^\dagger+\frac{\bar z}{\sqrt 2} s)(\bfa - \frac{z}{\sqrt 2}s) - \bfa^\dagger(\bfa - \frac{z}{\sqrt 2} s) \Big)e^{s\bfq_z - \frac{\coth(\beta/2)} 4 s^2}\\
&\quad\,  + e^{-\beta/2}\Big(  \frac 12 \bfa\bfa^\dagger +\frac 12(\bfa - \frac{z}{\sqrt 2}s)(\bfa^\dagger+\frac{\bar z}{\sqrt 2} s) - \bfa(\bfa^\dagger + \frac{\bar z}{\sqrt 2} s) \Big)e^{s\bfq_z - \frac{\coth(\beta/2)} 4 s^2}\\
& = e^{\beta/2}\Big(\frac 12 s\bfq_z-\frac 14 s^2 \Big)e^{s\bfq_z - \frac{\coth(\beta/2)} 4 s^2}+e^{-\beta/2}\Big(-\frac 12 s\bfq_z-\frac 14 s^2 \Big)e^{s\bfq_z - \frac{\coth(\beta/2)} 4 s^2}\\
& = \frac{e^{\beta/2}-e^{-\beta/2}}{2} s\Big(  \bfq_z -\frac{\coth(\beta/2)} 2  s  \Big)e^{s\bfq_z - \frac {\coth(\beta/2)} 4 s^2}\\
& = \frac{e^{\beta/2}-e^{-\beta/2}}{2} s\frac{\dd}{\dd s} e^{s\bfq_z - \frac \nu4 s^2}\\
&= \frac{e^{\beta/2}-e^{-\beta/2}}{2} \sum_{k=0}^\infty \frac{s^{k}}{(k-1)!}h_n(\bfq_z).
\end{align*}
Comparing to~\eqref{eq:hermit} we find that 
$$\mathcal L^*\big(h_k(\bfq_z)\big) =  \big(\sinh(\beta/2)k\big)\,  h_k(\bfq_z).$$
Finally, note that the operators $h_k(\bfq_z)$, over all $k\geq 0$ and $|z|=1$, span the whole space of polynomials of $\bfa$ and $\bfa^\dagger$, which is dense in $L_2(\sigma)$.

\end{document}